\documentclass[11pt]{article}

\usepackage{amsmath}
\usepackage[notheorems]{definitions}
\usepackage{fullpage}
\usepackage{dsfont}
\usepackage{ifthen}
\usepackage{xspace}
\usepackage{color-edits}

\addauthor{bl}{red}
\addauthor{dk}{blue}

\newcommand{\ignore}[1]{}

\newtheorem{theorem}{Theorem}[section]

\newtheorem{claim}[theorem]{Claim}

\newtheorem{proposition}[theorem]{Proposition}

\newtheorem{definition}{Definition}
\newtheorem{example}{Example}

\newcommand{\vc}[1]{\ensuremath{\mathbf{#1}}\xspace}
\newcommand{\INDICATOR}{\ensuremath{\mathds{1}}\xspace}
\newcommand{\Indicator}[1]{\ensuremath{\INDICATOR[#1]}\xspace}

\newcommand{\NUMENG}{\ensuremath{K}\xspace}
\newcommand{\NUMPAGES}{\ensuremath{N}\xspace}

\newcommand{\PageSet}{\ensuremath{\Omega}\xspace}
\newcommand{\DF}{\ensuremath{\beta}\xspace}
\newcommand{\TypeDist}{\ensuremath{\Gamma}\xspace}
\newcommand{\TypeSupp}{\ensuremath{\text{Supp}(\TypeDist)}\xspace}
\newcommand{\Td}[1]{\ensuremath{\gamma(#1)}\xspace}

\newcommand{\SelRule}[1][]{\ensuremath{%
\ifthenelse{\equal{#1}{}}{f}{f_{#1}}}\xspace}

\newcommand{\SatProb}[1]{\ensuremath{q_{#1}}\xspace}
\newcommand{\SatProbs}{\ensuremath{\vc{q}}\xspace}
\newcommand{\SatProbV}[1]{\ensuremath{\vc{q}_{#1}}\xspace}
\newcommand{\SatProbPage}[1]{\ensuremath{\vc{q}(#1)}\xspace}
\newcommand{\SP}[2]{\ensuremath{q_{#1}(#2)}\xspace}
\newcommand{\SPr}[1]{\ensuremath{q(#1)}\xspace}
\newcommand{\SatMat}{\ensuremath{\mathbf{Q}}\xspace}

\newcommand{\PERM}{\ensuremath{\sigma}\xspace}
\newcommand{\Matching}[1][]{\ensuremath{%
\ifthenelse{\equal{#1}{}}{\PERM}{\PERM_{#1}}}\xspace}
\newcommand{\MatchingS}[2][]{\ensuremath{%
\ifthenelse{\equal{#1}{}}{\PERM[#2]}{\PERM_{#1}[#2]}}\xspace}

\newcommand{\Perm}[1][]{\ensuremath{%
\ifthenelse{\equal{#1}{}}{\PERM}{\PERM_{#1}}}\xspace}
\newcommand{\PermS}[2][]{\ensuremath{%
\ifthenelse{\equal{#1}{}}{\PERM[#2]}{\PERM_{#1}[#2]}}\xspace}
\newcommand{\Perms}{\ensuremath{\vc{\PERM}}\xspace}

\newcommand{\OptPerm}[1][]{\ensuremath{%
\ifthenelse{\equal{#1}{}}{\hat{\PERM}}{\hat{\PERM}_{#1}}}\xspace}
\newcommand{\OptPermS}[2][]{\ensuremath{%
\ifthenelse{\equal{#1}{}}{\hat{\PERM}[#2]}{\hat{\PERM}_{#1}[#2]}}\xspace}
\newcommand{\OptPerms}{\ensuremath{\vc{\hat{\PERM}}}\xspace}

\newcommand{\PermDist}[2][]{\ensuremath{%
\ifthenelse{\equal{#1}{}}{\mathcal{D}_{#2}}{\mathcal{D}_{#2}^{(#1)}}}\xspace}
\newcommand{\PermDists}{\ensuremath{\vc{\mathcal{D}}}\xspace}
\newcommand{\MatchDist}[2][]{\PermDist[#1]{#2}}

\newcommand{\JOINTDENSITY}{\ensuremath{\rho}\xspace}
\newcommand{\JointDensity}[2][]{\ensuremath{%
\ifthenelse{\equal{#1}{}}{\JOINTDENSITY(#2)}{\JOINTDENSITY_{#1}(#2)}}\xspace}
\newcommand{\MarginalDensity}[1]{\ensuremath{\JOINTDENSITY_{#1}}\xspace}
\newcommand{\MixedPermDists}{\ensuremath{\mathbf{\mathcal{C}}}\xspace}

\newcommand{\NASH}{\ensuremath{\mathcal{N}}\xspace}

\newcommand{\CORREQSET}{\ensuremath{\text{EQ}^{\text{Corr}}}\xspace}
\newcommand{\CorrEqSet}[2]{\ensuremath{\CORREQSET_{#1}(#2)}\xspace}
\newcommand{\CorrEq}{\MixedPermDists}
\newcommand{\SIGNAL}[1][]{\ensuremath{%
\ifthenelse{\equal{#1}{}}{Z}{Z_{#1}}}\xspace}
\newcommand{\SIGNALV}{\ensuremath{\vc{Z}}\xspace}

\newcommand{\Payoff}[2]{\ensuremath{u_{#1}(#2)}\xspace}
\newcommand{\Welfare}[2][]{\ensuremath{%
\ifthenelse{\equal{#1}{}}{W(#2)}{W_{#1}(#2)}}\xspace}

\newcommand{\SearchGame}[2]{\ensuremath{\textsc{Search}_{#1}(#2)}\xspace}
\newcommand{\SearchGameName}{Search Engine Competition Game\xspace}
\newcommand{\SingletonGameName}{Singleton Search Engine Game\xspace}

\newcommand{\StatProb}[1][]{\ensuremath{\ifthenelse{\equal{#1}{}}{\mathbf{\pi}}{\pi_{#1}}}\xspace}
\newcommand{\RetFail}[1]{\ensuremath{R^{(f)}_{#1}}\xspace}
\newcommand{\RetSucc}[1]{\ensuremath{R^{(s)}_{#1}}\xspace}
\newcommand{\TransFail}[2]{\ensuremath{\tau^{(f)}_{#1,#2}}\xspace}
\newcommand{\TransSucc}[2]{\ensuremath{\tau^{(s)}_{#1,#2}}\xspace}
\newcommand{\ExitFail}[1]{\ensuremath{p^{(f)}_{#1}}\xspace}
\newcommand{\ExitSucc}[1]{\ensuremath{p^{(s)}_{#1}}\xspace}

\newcommand{\OPT}[1][]{\ensuremath{%
\ifthenelse{\equal{#1}{}}{\text{OPT}}{\text{OPT}_{#1}}}\xspace}
\newcommand{\EQ}[1][]{\ensuremath{%
\ifthenelse{\equal{#1}{}}{\text{EQ}}{\text{EQ}_{#1}}}\xspace}

\begin{document}


\begin{titlepage}
\title{User Satisfaction in Competitive Sponsored Search}

\author{David Kempe\thanks{Department of Computer Science, University of Southern California}
\and
Brendan Lucier\thanks{Microsoft Research, New England}
}
\date{}

\maketitle

\begin{abstract}
We present a model of competition between web search algorithms, and study the impact of such competition on user welfare. In our model, search providers compete for customers by strategically selecting which search results to display in response to user queries. Customers, in turn, have private preferences over search results and will tend to use search engines that are more likely to display pages satisfying their demands.

Our main question is whether competition between search engines increases the overall welfare of the users (i.e., the likelihood that a user finds a page of 
interest). When search engines derive utility only from customers to whom they show relevant results, we show that they differentiate their results, and every equilibrium of the resulting game achieves at least half of the welfare that could be obtained by a social planner. This bound also applies whenever the likelihood of selecting a given engine is a convex function of the probability that a user's demand will be satisfied, which includes natural Markovian models of user behavior.

On the other hand, when search engines derive utility from all customers (independent of search result relevance) and the customer demand functions are not convex, there are instances in which the (unique) equilibrium involves no differentiation between engines and a high degree of randomness in search results.  This can degrade social welfare by a factor of $\Omega(\sqrt{\NUMPAGES})$ relative to the social optimum, where \NUMPAGES is the number of webpages.  These bad equilibria persist even when search engines can extract only small (but non-zero) expected revenue from dissatisfied users, and much higher revenue from satisfied ones.

\end{abstract}
\end{titlepage}




\section{Introduction} \label{sec:introduction}

One of the most central tenets of economic market theory is that competition
between vendors will benefit consumers, whether by driving down prices,
encouraging product differentiation, or increasing the overall quality
of goods.
This belief is supported by many models and a long
history of empirical studies.
(See, e.g., \cite{mas-collel:whinston:green,varian:micro} for classic overviews.)

In a market where consumers directly buy products from vendors, these
benefits of competition are clear: given a choice between multiple 
competing products, a customer can be expected to choose whichever he
finds most appealing at its given price. 
This dynamic naturally leaves inferior and overpriced products unsold.

In this paper, we analyze a game of competition 
between \emph{search engines},
and its impact on the welfare of the users.
The main difference between the market for sponsored search and
classic oligopoly models is that in sponsored search, users are not the
\emph{customers} (since search results are provided to them for free), 
but the \emph{product}, which is sold to advertisers.
As a result, a search engine might have an interest in satisfying users
only inasmuch as it can monetize their visit and satisfaction --- this
results in a potentially serious misalignment of incentives, and thus
loss in user welfare.
Of course, this effect may be partially mitigated by the fact that
a reputation for quality can help to attract future usage of a search
engine, so competition may help users.

At a high level, a competition between search engines is essentially a
competition between (potentially randomized) \emph{algorithms} that 
return collections of webpages in response to search queries. 
A query typically does not completely determine which
web pages are being searched for.\footnote{Think of polysemy, or
simply different information needs of different users.}
As a result, search engines will generally not be able to
deterministically satisfy every user;
in return, a searcher may not know whether a set of results
will satisfy his demand before choosing a search
provider.\footnote{In this sense, search results are an ``experience
  good'' \cite{shapiro:varian:information-rules}.}

If there is only a single engine that acts as a monopolist, the
algorithmic solution to the web search problem is clear: 
given a query, the engine should display the pages that maximize the
probability that the user is satisfied. 
(We assume that engines have accurate estimates of searchers'
  probabilities for desiring different pages.)
When multiple engines compete for customers, it is less clear how the
engines will behave at equilibrium. 
Will they differentiate their search results to capture
different parts of the market, or will competition perversely drive
them toward displaying esoteric webpages in their rankings, degrading
the overall user experience?

The main insight of our work is that
the answer depends on two features of the competition:
(1) the alignment of utility between the engines and users, and 
(2) how users respond to search result quality.
In order to elaborate on these insights, we briefly describe some
features of our basic model of competitive search.
(The missing parts of the model are defined in Section~\ref{sec:models}.)

\subsection{A brief overview of the model}
\NUMENG search engines are competing for the attention of users, by
displaying the set of \NUMPAGES web pages in some order.
Each query by a user\footnote{whom we can imagine as drawn
    from an infinite population}
induces a known probability distribution 
over subsets $S$ of web pages; the probability associated with
the set $S$ is the probability that the user would be satisfied upon
seeing at least one page in $S$.
When the query arrives, each search engine can select (a
distribution over) the order in which the \NUMPAGES pages are displayed.
In addition to the set $S$, the user's type also includes a patience
threshold $t$: he is satisfied if the first $t$ pages in the ordering
contain at least one page from $S$.

Based on the distributions of orders in which the \NUMENG search engines
display pages, a user interested in a set $S$ will
choose an engine to visit.  
The choice can be probabilistic, and will be correlated in some way
with search performance.
We assume that the probability of choosing engine $i$ is monotone
non-decreasing in the probability that $i$ would satisfy the user, and
monotone non-increasing in the probabilities that the competing engines
$i' \neq i$ would satisfy him. 
This model of user behavior is very general, allowing for many
concrete instantiations of such \emph{selection rules}, 
including natural Markovian models\footnote{For
example, a user might stick with his current engine until it fails to
satisfy a query, at which point he switches to an alternative.},
proportional choice, and others.

We note that even though a user's choice depends on the search engines'
algorithms, the intended interpretation is not that users know their desired
webpages in advance (which would eliminate the need for search engines).
Rather, our model is intended to capture the fact that 
the fraction of users choosing an engine can depend, possibly in the long
run, on the likelihood that a query on a certain topic will result in
satisfactory search results.

It remains to describe the search engine payoffs in the game.
When a user visits a search engine, it derives
revenue by showing ads to the user.  In sponsored search, the user
is shown ads that are aligned with his perceived intent.  Thus, even
conditioned on a user having chosen a given engine, the revenue generated
can be correlated with how well the engine inferred the user's intent, which
is correlated with whether the user was satisfied by the search results.
Motivated by this, we define engine utilities as follows: an engine receives
a payoff of $1$ for every user that visits the engine \emph{and is satisfied},
and a payoff of $\DF \in [0,1]$ for users that visit the engine but are
not satisfied.  
That is, when $\DF = 1$ an engine receives full benefit from each
user regardless of the search results, whereas
$\DF = 0$ implies that an engine receives no revenue from a customer
if it presents irrelevant results.

There is a natural interpretation of the parameter \DF in
terms of the type of advertising chosen by the search engines.
When using per-impression advertising (e.g., banner ads), the search
engine receives the full payment from the advertiser when a user
visits, regardless of satisfaction; this corresponds to $\DF=1$.
On the other hand, with per-click advertising, the search engine is
paid only when a user clicks on an ad, which will be highly
correlated with the relevance of the results and ads displayed; thus,
pure per-click advertising corresponds to $\DF=0$.
A mix of the two advertising models naturally gives rise to
intermediate values of \DF.

We emphasize that even in the extreme case of $\DF=1$ (a fully
per-impression advertising model), user satisfaction is still relevant
as it provides a way to attract users away from competing engines. 
The difference lies in how a user is valued, \emph{given} that he has
already made his choice about which engine to use.

\subsection{Our Results}

We find that the outcome of competition depends heavily 
on both the value of \DF (i.e., on whether or not an engine derives direct
benefit from unsatisfied users), 
and on the form of the selection rule (i.e.,  how users respond to
changes in satisfaction probability).
Our main analytical result, informally stated, is the following dichotomy:
\begin{enumerate}
\item If $\DF = 0$ (search engines obtain revenue only when a
user is satisfied), or if the engine selection rule is convex, 
then the equilibrium (which is essentially unique)
will have a large amount of specialization. 
Engines will choose deterministic strategies aimed at
targeting particular subsets of the users.
Regardless of the number of search engines, the Price of
Anarchy is bounded by 2, even for mixed or correlated equilibria,
and this bound is tight.
In other words, competition degrades user satisfaction at most by a
factor 2 compared to the outcome if engines were to pool their
resources to best satisfy users.

\item In contrast, we show in Section~\ref{sec:beta1} that
if $\DF > 0$, then for a large class of non-convex engine selection
rules, there exist symmetric equilibria, 
i.e., all engines use the same algorithm.
In this equilibrium, there is no specialization, meaning that users
are no better off than with a single engine.
In fact, there are cases in which the 
equilibria can be \emph{worse} than the optimal
single-engine solution, with a pure Price of Anarchy
of $\Omega(\sqrt{\NUMPAGES})$ even when $\NUMENG = 2$.
When $\DF = 1$ we show that this equilibrium is unique, and hence
even the pure Price of Stability is $\Omega(\sqrt{\NUMPAGES})$.
\end{enumerate}

Thus, informally stated, when search engines only want to
satisfy viewers (e.g., because they use purely per-click advertising), 
searcher welfare is guaranteed to be close to the social optimum,
whereas it can be very far from optimum when search engines only want
to attract viewers (e.g., because they use purely
  per-impression advertising). 
However, these inefficiencies vanish under certain models of user behavior
that result in convex selection rules.  The intuition behind this phenomenon 
is that convexity incentivizes engines to have high satisfaction 
probability for a few user types, rather than infrequently satisfying
many types.
The overall result is specialization among the engines if
  there are enough high-demand pages.
As a concrete example, we show that convexity occurs under a natural
class of Markovian models of user behavior.

\subsection{Related Work}
\label{sec:related-work}
Our model is closely related to the study of competitive algorithm
design due to Immorlica et al.~\cite{dueling}.
In that work, two algorithms (for a general optimization problem) compete 
to maximize the probability that they will outperform their opponent
on an (unknown, randomly selected) input.
Competitive search is given as an example in \cite{dueling}, though
their model differs significantly from our own: in their version, the
goal of each engine is to place a single desired page closer to the
top of the search results than the opponent.
The all-or-nothing nature of this goal creates some unintuitive
best-response strategies, and may be questionable in terms of modeling
actual web searcher behavior.
In contrast, we focus on a richer model of competition between 
search algorithms, and explore the (in)efficiency of equilibria for this 
class of games.
There does not appear to be a direct reduction from the model of 
\cite{dueling} to our general model.

A somewhat similar model of search engine competition
was proposed by Khoussainov and Kushmerick
\cite{khoussainov:thesis,khoussainov:kushmerick}.
They consider a metasearcher that directs queries to
specialized engines.
The decision space and utility structure of the game are different:
Khoussainov and Kushmerick assume that search engines decide on topics
to include into (or exclude from) their indexes, and that there is a cost
associated with including topics. 
The selection rule used in their model is the Majority Selection
Rule from Section~\ref{sec:selection-rules} of our paper,
and a search engine derives full benefit from all users that are
directed to it (in our terminology, the case $\DF = 1$). 
Rather than studying the equilibria of the resulting game,
Khoussainov and Kushmerick
\cite{khoussainov:thesis,khoussainov:kushmerick}
propose and experimentally evaluate reinforcement learning
approaches for the search engines.

Differentiation between competing search engines has been
studied under economic models by Mukhodhyay, Telang, and
Rajan~\cite{mukhopadhyay:telang:rajan}. They primarily focus on
\emph{vertical} differentiation, i.e., why lower-quality search
engines can survive in a market.
Mukhodhyay, Telang, and Rajan
also include a model including horizontal differentiation, which
means that search engines can choose different answers to appeal to
different users; their models are stylized in that they allow
choice of a quality level and a ``location'' for the engine; as such,
they somewhat resemble Hotelling's model for competition
\cite{hotelling:stability}.

Competition between search engines has also been studied in an
economic model by Argenton and Pr\"{u}fer
\cite{argenton:prufer:competition-externalities}. Their model focuses
on externalities that arise when engines can use data from past
searches to improve future searches, leading to herding behavior
toward a monopoly. Their main recommendation is that in order to
encourage competition, engines should be forced to share search
histories with each other. In this sense, their model focuses on an
aspect orthogonal to ours.

Competition for user participation has been modeled in the
context of auctions
\cite{liu:chiu:competition,mcafee:comp.sellers,peters:comp.auction}.
In these settings, each auctioneer selects a mechanism;
bidders then select an auction to participate in
and decide on their bidding strategy. 
The nature of payoffs to the players is significantly different
between auctioneers and search engines; among others, the participants
(bidders) compete with each other after choosing an auction, and are
fully rationally in their choices, whereas we do not assume
rationality of the searchers, and searchers naturally do not compete
with each other.

The revenue of a single search engine has been studied 
under various models of user behavior and advertising revenue
\cite{aggarwal:feldman:muthukrishnan:pal,athey:ellison:position-auctions,edelman:ostrovsky:schwarz,ClickCascades}.
These works focus on the game played by
advertisers, rather than by multiple search engines.

The analysis of product differentiation in competitive markets has a
long history in economic theory.
The classic Hotelling model \cite{hotelling:stability} 
implies that competing firms will tend to produce similar products.
Later works, beginning with Aspremont, Gabszewicz, and Thisse
\cite{aspremont:hotelling}, instead find that product differentiation
can improve revenue by, for example, thinning markets to affect future
prices.  
See \cite{varian:micro} for a treatment of this classic theory.
In contrast, differentiation occurs in our model as a way for
competitors to reduce customer uncertainty about the product
they will receive from each.
Correspondingly, this differentiation occurs at equilibrium when the
incentives of the search engines are strongly aligned with the customers'
ultimate product satisfaction.

Our analysis of competition in the case $\DF = 0$ shares some similarities
with the analysis of valid utility games \cite{vetta:nash} and smooth
games \cite{roughgarden:smoothness}.
We note that the click-through variant of our game is mathematically
similar to Oren and Kleinberg's recent analysis of a research credit
allocation game \cite{kleinberg:misalloc}, where researchers select
projects in which to invest effort.
Besides the motivation behind these games, the primary difference is
that we model a search engine's (pure) strategy as a probability
distribution over search results, rather than a deterministic list of
pages.





\section{Models and Concepts} 
\label{sec:models}

We employ the following (standard) notation.
We write $[n] = \SET{1, \ldots, n}$. 
Vectors are denoted by bold face.
For a vector $\vc{x}$, we use $\vc{x}_{-i}$ to denote the vector
without coordinate $i$, and $(y, \vc{x}_{-i})$ for the vector with 
coordinate $i$ replaced by $y$.
We denote the indicator function by \INDICATOR; that is,
for an event \Event{E}, we write $\Indicator{\Event{E}}$ for the
function that is $1$ when \Event{E} happens and $0$ otherwise.

We will frequently reason about distributions over discrete sets (or
their power sets).
Let $\mathcal{A}$ be a finite set,
and $p$ a distribution over $\mathcal{A}$.
We say that $p$ \emph{is in general position}
if for every pair $S,S' \subseteq \mathcal{A}$ of distinct sets 
($S \neq S'$), we have that $p(S) \neq p(S')$.
Notice that the set of distributions not in general position has
measure 0, i.e., distributions are generically in general position.

\subsection{The \SearchGameName}
\label{sec:game-setup}
Our formal model of competition between search engines is as follows.
The game is played by \NUMENG engines.
There is a universe \PageSet of \NUMPAGES possible webpages, 
and an infinite population of \emph{users} (or 
\emph{searchers}).\footnote{We consistently use female pronouns for search
engines and male pronouns for searchers.}
Users wish to find webpages, and search engines are
interested in attracting users, and possibly also in satisfying their
search needs.

Each user has a private \emph{type} specified by a subset 
$S \subseteq \PageSet$ of pages and a slot threshold $t \geq 0$.
$S$ is the set of pages that would satisfy the user's request, 
and we call it the \emph{desired pages} or \emph{satisfying pages} of the user.
$t$ is the number of search results the user will examine to
find a satisfying page before abandoning the search; we also call $t$
the \emph{patience threshold} (or simply \emph{threshold}) of the user.

There is a commonly known distribution \TypeDist over user types,
where \Td{S,t} is the probability of drawing the pair $(S,t)$.
This distribution incorporates the search query
issued by the user; in other words, we consider the game for a
specific query, since games for different queries can be solved
independently.\footnote{This holds under the (strong)
    assumption that users choose their engine per query.
    See a brief discussion of this issue in Section~\ref{sec:conclusions}.}
Without loss of generality, we assume that 
$\sum_{S \ni n} \sum_t \Td{S,t} > 0$ for all pages $n \in \PageSet$,
i.e., each page satisfies at least one type in the support
of \TypeDist.

The search engines have knowledge of \TypeDist.
Their strategies are based on \emph{permutations} of pages \dkedit{$j \in \PageSet$.}
When engine $i$ plays permutation \Perm[i],
she places page $\Perm[i](j)$ in slot $j$ (for each $j$).
We also think of \Perm[i] as a \emph{matching} between
slots and pages,
and write $\PermS[i]{t} = \Set{\Perm[i](j)}{j \in [t]}$ for the set
of pages placed in the first $t$ slots by engine $i$.
A user with type $(S,t)$ is \emph{satisfied} by a matching
$\Perm[i]$ precisely if $\PermS[i]{t} \cap S \neq \emptyset$,
i.e., if at least one of the pages satisfying the user is among the
top $t$ pages displayed by the engine.

When a user visits engine $i$, the engine obtains utility $1$ if the user
is satisfied, and utility $\DF \in [0,1]$, a fixed parameter of
the game, if the user is not satisfied. 
The parameter \DF models the tradeoff between optimizing for
market share vs.~customer satisfaction, 
or the fraction of display ads among the ads shown by the engine.
A large value of \DF means that the engine
would simply like to attract many users, while a small value implies
that users are primarily (or only) valuable if they are satisfied.

A \emph{pure strategy} for engine $i$ is a distribution
$\PermDist{i}$ over matchings. 
(This apparent misnomer is discussed in more detail
in Section~\ref{sec:equilibria}.)
A \emph{(pure) strategy profile}
$\PermDists = (\PermDist{1}, \dotsc, \PermDist{\NUMENG})$ 
is a vector of strategies, one for each engine. 
A user type $(S,t)$ and engine strategy $\PermDist{i}$ together
determine a satisfaction probability 
$\SatProb{i}(S,t,\PermDist{i}) =
\Prob[\Perm \sim \PermDist{i}]{\PermS{t} \cap
 S \neq \emptyset}$.
We often write $\SatProb{i}$ for
$\SatProb{i}(S,t,\PermDist{i})$ when the type and strategy are clear
from the context.
The vector of the satisfaction probabilities for all engines is
denoted by $\SatProbs(S,t,\PermDist{i})$ or $\SatProbs$.
Engines may also mix strategies, playing distributions over
distributions. 
In particular, below, we will be interested in
correlated equilibria of the game.

Two permutations $\Perm, \Perm'$ are
\emph{equivalent} with respect to \TypeDist if, for every type $(S,t)$
in the support of \TypeDist, $(S,t)$ is satisfied by either both or
neither of the permutations.
For example, if every user has a threshold of at least $2$, then two
permutations that differ only in the order of the first two pages are
equivalent.  
A pure strategy is \emph{deterministic} if its support
consists entirely of permutations that are equivalent to each other; 
we abuse terminology by stating that such a
strategy consists of a single permutation. 
A strategy profile (pure or mixed) is \emph{deterministic} if the
support of each engine's distribution contains only
deterministic strategies.

A user of type $(S,t)$ must select which search engine to use,
taking into account the profile $\SatProbs$ of satisfaction probabilities.  
A \emph{selection rule} is a function \SelRule
which maps a profile of satisfaction probabilities
$\SatProbs$ to a
distribution over search engines.  We write $\SelRule[i](\SatProbs)$
for the probability that a user selects engine $i$ given $\SatProbs$.
We discuss properties of selection rules and several concrete examples
in more detail in Section~\ref{sec:selection-rules} below.

The precise sequence of events in the game is as follows.
First, the engines jointly determine a mixed strategy profile
\MixedPermDists. 
Let \JointDensity{\PermDists} be the density function of the
joint distribution of pure strategies under \MixedPermDists.
Notice that we allow the random choices to be correlated.
Once \MixedPermDists is determined, 
a vector of pure (but not necessarily deterministic)
strategies \PermDists is drawn according to \JOINTDENSITY.
Next, a user type $(S,t)$ is realized from \TypeDist.
\PermDists is known to the user,
and determines a satisfaction probability profile
$\SatProbs = (\SatProb{1}, \dotsc, \SatProb{\NUMENG})$ for the engines.
The user selects a search engine stochastically according to the
distribution $\SelRule(\SatProbs)$.
After the user selects an engine, say $i$, a permutation
$\Perm[i]$ is realized from the distribution $\PermDist{i}$.
If the user is satisfied by $\Perm[i]$, engine $i$ receives a
payoff of $1$; otherwise, she receives a payoff of \DF.
All other engines $i' \neq i$ receive payoff $0$. 
In summary, the payoff to engine $i$ under the joint strategy
  profile \MixedPermDists is 
\begin{align}
\Payoff{i}{\MixedPermDists}
& = 
\int_{\PermDists}
\sum_{\substack{S \subseteq \PageSet \\ t \geq 0}} 
\Td{S,t} 
\cdot \SelRule[i](\SatProbs(S,t,\PermDists)) \cdot
(\DF + (1-\DF)\SatProb{i}(S,t,\PermDist{i}))
\cdot 
\JointDensity{\PermDists} d\PermDists.
\label{eqn:engine-payoff}
\end{align}

We call this game the \emph{\SearchGameName}
with parameters $\DF$ and $\TypeDist$, 
and denote it by \SearchGame{\DF}{\TypeDist}.

\subsection{Strategies, Equilibria, and the Price of Anarchy}
\label{sec:equilibria}
We now discuss in more depth the choice to consider
distributions \PermDist{i} as \emph{pure} strategies.
The reason that we cannot simply consider engines as \emph{mixing}
over pure strategies which are \emph{permutations} is that the
engines' payoffs are defined in terms of user behavior, and users'
behaviors depend on the entire probability distributions chosen by the
engines. 

In other words, it is necessary to know the
distributions over pure strategies in order to even define the game's
payoffs. 
Thus, the pure strategies have to be the engines'
distributions \PermDist{i}.
For the same reason, when engines are playing mixed strategies, it is
necessary to assume that the searchers are aware of the pure
strategies drawn by the engines; this would be realistic if the
changes in engines' strategies occurred at a much slower time scale
than the users' adaptation.


We will carefully distinguish nomenclature between pure strategies and
\emph{deterministic} strategies, which --- as defined above --- are
ones that have probability 1 of selecting a particular permutation (or
its equivalence class).
With this understanding in hand, a \emph{(pure) Nash Equilibrium} of
the game is a strategy profile \PermDists in pure strategies from
which no single engine can unilaterally deviate to strictly improve
her utility.
Pure Nash Equilibria may not exist (even for $\DF=0$) ---
we show a simple example of this in 
Appendix~\ref{app:non-existence} ---
so we broaden the space of allowed strategy profiles.

\begin{definition}[Correlated Equilibrium]
Fix a finite probability space $(\Omega,\delta)$.
For each $i = 1, \ldots, \NUMENG$, let 
$\SIGNAL[i]: \Omega \to \mathcal{Z}_i$ be a 
\emph{signal} (random variable), using the same domain, and write
$\SIGNALV = (\SIGNAL[1], \ldots, \SIGNAL[\NUMENG])$.

The strategy of player $i$ is a mapping
$\mu_i$ from $\mathcal{Z}_i$ to the pure strategies of player $i$.
$(\mu_1, \mu_2, \ldots, \mu_{\NUMENG})$ is a 
\emph{correlated equilibrium} iff for each player $i$ and each mapping
$\mu'_i$ from $\mathcal{Z}_i$ to the pure strategies of player $i$, we
have that
\begin{align} \label{eq:correq}
\Expect[\SIGNALV \sim \delta]{\Payoff{i}{(\mu_1(\SIGNAL[1]), \ldots, \mu_i(\SIGNAL[i]), \ldots, \mu_{\NUMENG}(\SIGNAL[\NUMENG]))}}
& \geq
\Expect[\SIGNALV \sim \delta]{\Payoff{i}{(\mu_1(\SIGNAL[1]), \ldots, \mu'_i(\SIGNAL[i]), \ldots, \mu_{\NUMENG}(\SIGNAL[\NUMENG]))}}.
\end{align}
In other words, given the signal he received, each player prefers
playing according to the suggested strategy over any other strategy.
\end{definition}

We denote the set of all Correlated Equilibria by
\CorrEqSet{\DF}{\TypeDist}.
As usual, we omit the parameters when they are clear from the context,
in which case we simply write \CORREQSET.

We measure the quality of a strategy profile by the welfare
of the searchers, i.e., the probability that a user drawn from 
\TypeDist is satisfied by his chosen engine.
The welfare is
\begin{align}
\Welfare{\MixedPermDists}
& = 
\int_{\PermDists}
\sum_{\substack{S \subseteq \PageSet\\ t\geq 0}} 
\Td{S,t} \cdot \sum_i \SelRule[i](\SatProbs(S,t,\PermDists)) \cdot
\SatProb{i}(S,t,\PermDist{i})
\cdot \JointDensity{\PermDists} d\PermDists.
\label{eqn:searcher-utility}
\end{align}

It is well known that equilibria can sometimes be very
inefficient in terms of the welfare they yield.
Let \OPT denote a strategy profile for the engines that 
maximizes the searcher welfare.
The degradation of welfare at equilibrium, compared to this optimum
profile, is typically measured by two quantities:
the (pure, mixed, or correlated) Price of Anarchy (PoA)
\cite{koutsoupias:papadimitriou:anarchy,roughgarden:smoothness}
$\sup_{\PermDists \in \CORREQSET} \frac{\Welfare{\OPT}}{\Welfare{\PermDists}}$
captures the worst-case welfare loss resulting from selfish
rational behavior by the engines;
the Price of Stability (PoS) \cite{ADKTWR}
$\inf_{\PermDists \in \NASH} \frac{\Welfare{\OPT}}{\Welfare{\PermDists}}$
captures the welfare loss incurred by insisting that a strategy
profile must be stable against selfish unilateral deviations.
(Here, the infimum is usually taken over Nash Equilibria \NASH, as
correlated equilibria are more permissive, and would thus lower the
value.)

For our upper bounds on the Price of Anarchy, we will prove
bounds on the \emph{Correlated} Price of Anarchy,
while all our lower bounds will be achieved by pure Nash Equilibria.

\subsection{The \SingletonGameName}
\label{sec:singleton-game}
All of the lower bounds on PoA (or PoS)
which we prove in this paper arise already in a very restricted form
of the \SearchGameName, which we call the \SingletonGameName.
In the \SingletonGameName, 
each user always has a patience threshold of $1$, and desires a
singleton set. Thus, the distribution \TypeDist can be specified
simply as $\Td{n} := \Td{1, \SET{n}}$.
Furthermore, since no user ever looks past the first slot, engine
strategies are equivalent iff they agree in the page in the first
slot. In particular, for the purpose of the game, a distribution
\PermDist{} is completely specified by the probability distribution of
pages in slot 1. We will therefore frequently speak simply of an
engine ``displaying'' or ``choosing'' a page, without specifying a
slot (which is understood to be slot 1).

This distribution is characterized by the probability of
satisfying any given type of user (which can now be identified
with the desired page $n$).
Therefore, for instances of the \SingletonGameName, an engine $i$'s
pure strategy can be fully specified by the probability vector
$\SatProbV{i} = (\SP{i}{n})_{n=1}^N$ of probabilities of displaying
pages $n=1, \ldots, N$.
In this representation, the utility of engine $i$ under a pure
strategy profile is
$\sum_{n \in \Omega} \Td{n} \SelRule[i](\SP{i}{n})
\cdot (\DF + (1-\DF)\SP{i}{n})$.

\subsection{User Behavior and Selection Rules}
\label{sec:selection-rules}
We model the users as stochastic entities, choosing engines via a
probabilistic process that depends on their probability of being
satisfied.  
A selection rule $\SelRule$ captures a
model of user behavior.  We assume the same selection rule is
used by each user,\footnote{We note that the space of selection rules
is convex, so the selection rule in which each user independently randomizes
between using a rule $\SelRule_1$ or an alternative rule $\SelRule_2$
is, itself, a selection rule.} 
and that the effect of the engines' strategies on
user behavior can be captured entirely by the satisfaction probability
profile $\SatProbs$.  
We assume that each $\SelRule[i]$ is monotone non-decreasing in
$\SatProb{i}$ and non-increasing in $\SatProb{i'}$ for each $i' \neq i$.  
That is, an engine's probability of being chosen by a user of type
$(S,t)$ is non-decreasing in the probability that this user will
be satisfied, and non-increasing in the probability that other engines
will satisfy him.
We do not assume that $\SelRule[i]$ is differentiable
or continuous.

This model is very general, encompassing almost all rational
selection rules one could imagine. 
We illustrate it with several examples, 
capturing different models of behavior.

\begin{description}
\item[Proportional Selection]
\emph{Proportional selection} is defined by
$\SelRule[i](\SatProbs) = \frac{\SatProb{i}}{\sum_j \SatProb{j}}$, 
with $\SelRule[i](\mathbf{0}) := 1/\NUMENG$.
In this rule, a user picks engine $i$ with probability
proportional to the probability that he will be satisfied by engine $i$. 
Similar rules have frequently been used in past models that
involved a user selecting a vendor according to the quality of the
vendors' products.

\item[Markovian selection]
The \emph{Markovian selection rule} has
$\SelRule[i](\SatProbs) = \frac{b_i}{\sum_j b_j}$,
where $b_j = \frac{1}{1-\SatProb{j}}$ for each engine $j$.
If any $\SatProb{j}$ is equal to $1$, then 
$\SelRule[i](\SatProbs) = 1/k$ if $\SatProb{i} = 1$, where $k$ is the
number of engines with $\SatProb{j} = 1$, 
and $\SelRule[i](\SatProbs) = 0$ if $q_i < 1$.
This selection rule captures the stationary distribution of a process
in which a user repeatedly uses a given engine until she does not
satisfy him, then switches to another engine uniformly at random. 
(A more general version is defined in Section~\ref{sec:markov}.)

\item[Majority selection rule]
The \emph{Majority selection rule} is defined as follows:
Let $A = \argmax_i \SatProb{i}$ be the set of engines maximizing the
satisfaction probability of a given user.
Then the Majority selection rule sets $\SelRule[i](\SatProbs) = 1/\SetCard{A}$ for all $i \in A$,
and $\SelRule[i](\SatProbs) = 0$ for all $i \notin A$.
The Majority selection rule captures the behavior of a user who
chooses an engine that maximizes his probability of being satisfied,
breaking ties uniformly at random.
\end{description}

\section{Convexity, Determinism, and High Welfare} 
\label{sec:convexity}

In this section, we examine the role of convexity in the search
engines' utility functions. 
Roughly speaking, we show that when the selection rule \SelRule is
convex, each engine's utility is strictly convex.
This convexity, or alternatively, the assumption that $\DF=0$,
implies that every correlated equilibrium is deterministic;
in other words, engines do not mix over distributions of permutations,
but only over permutations.
For pure equilibria, this implies that engines simply play one
permutation (or an equivalence class of permutations)
deterministically.
In turn, we will show that the restriction to deterministic strategies
implies a bound of 2 on the correlated PoA.

In Section~\ref{sec:markov}, we will show that, for a natural
and general class of Markovian user behaviors, the selection rule
\SelRule is indeed convex; hence, for arbitrary \DF, this type of user
behavior yields a correlated PoA of at most 2.

We begin by defining a non-indifference property of selection rules.
Roughly speaking, we wish to exclude selection rules for which a
search engine's probability of satisfying a user has effectively no
impact on that user's selection.
We define the \emph{support} $S(\SatProbs)$ of a satisfaction
probability vector $\SatProbs$ to be the set of engines $i$ for which
$\SatProb{i} > 0$.   
We say that $\SelRule$ is \emph{non-indifferent} if the following two
properties hold.  First,
whenever $\SetCard{S(\SatProbs)} \geq 2$ 
and $\SatProb{j} < 1$ for some $j \in S(\SatProbs)$, 
there exists some $i \in S(\SatProbs)$ such that
$\SelRule[i](\SatProbs) < \SelRule[i](1, \SatProbs_{-1})$.  
That is, if at least two engines would satisfy a user with positive
probability, and not all supporting engines satisfy the user with
probability $1$, then at least one supporting engine would be more
likely to attract the user by sufficiently increasing her probability
of satisfying him. 
Second,
if there exists some $i$ such that $\SatProb{i} = 1$, then 
$\SelRule[j](\SatProbs) = 0$ for all $j \not\in S(\SatProbs)$.
That is, if some engine is guaranteed to satisfy the user, then
the user will not choose an engine that is guaranteed \emph{not}
to satisfy him.
The non-indifference property is quite weak; in particular, it is
significantly weaker than assuming that $\SelRule[i]$ is strictly
increasing in $\SatProb{i}$.
For example, both the Majority and Markovian selection rules described
at the end of Section~\ref{sec:models} are not strictly increasing in the
satisfaction probabilities, but do satisfy non-indifference. 

%
%

\begin{theorem}
\label{thm:convex.deterministic}
Suppose that \TypeDist is in general position, and at least one of the
following three holds for each $i$:
\begin{enumerate}
\item $\DF < 1$, and $\SelRule[i](\SatProb{i}, \SatProbs_{-1})$
 is a convex function of \SatProb{i};
\item $\DF = 1$, and $\SelRule[i](\SatProb{i}, \SatProbs_{-1})$
 is strictly convex in \SatProb{i};
\item $\DF = 0$ and $\SelRule[i]$ is non-indifferent.
\end{enumerate}

Then, each correlated equilibrium of the game
\SearchGame{\DF}{\TypeDist} is deterministic.
\end{theorem}

\begin{proof}
%
The high-level idea of the proof is the following:
if selection rules are convex or $\DF = 0$, an engine 
can increase her utility by choosing two outcomes she randomizes between,
deciding which one contributes more to the overall utility, and increasing the
probability of returning that outcome at the expense of the other.
The exact calculations are somewhat cumbersome.

Fix a correlated equilibrium \CorrEq, and assume for contradiction
that the support of \CorrEq is not entirely deterministic.
Fix an engine $i$ and a signal \SIGNAL[i] given to $i$ with positive
probability, such that conditioned on receiving \SIGNAL[i], engine $i$
plays a non-deterministic strategy \PermDist{i}.
From engine $i$'s viewpoint, conditioned on \SIGNAL[i], the other
engines $i' \neq i$ are playing mixed strategies, randomizing in
possibly correlated ways over possibly non-deterministic strategies.
Let $\JointDensity[-i]{\PermDists_{-i}}$ denote the joint density over the
strategies played by engines $i' \neq i$, conditioned on \SIGNAL[i].

Let \TypeSupp denote the support of \TypeDist, and consider a type
$(S,t) \in \TypeSupp$.
When engine $i$ satsifies users of type $(S,t)$ with probability $x$, 
its expected utility conditioned on the user type being $(S,t)$ is
\begin{align*}
g_i(x; S,t)
& = 
\int_{\PermDists_{-i}} 
\SelRule[i]((x, \SatProbs_{-i}(S,t,\PermDists_{-i}))) \cdot
(\DF + (1-\DF) \cdot x) \cdot
\JointDensity[-i]{\PermDists_{-i}} d \PermDists_{-i}.
\end{align*}

If $\SelRule[i]((x, \SatProbs_{-i}(S,t,\PermDists_{-i})))$ is a convex
function of $x$ and $\beta < 1$, 
or if $\SelRule[i]$ is strictly convex in $x$, then for each $(S,t)$, 
$\SelRule[i]((x, \SatProbs_{-i}(S,t,\PermDists_{-i}))) \cdot
(\DF + (1-\DF) \cdot x)$ is a strictly convex function of $x$, and
hence so is $g_i(x; S,t)$.
Therefore, under either of the first two assumptions of the theorem, 
$g_i(\cdot; S,t)$ is strictly convex.
(The analysis in the third case $\DF=0$ will proceed similarly, but
not explicitly be based on strict convexity.)

The utility of engine $i$ from playing at this equilibrium is
\begin{align*}
\Payoff{i}{\PermDists}
& = 
\sum_{S \subseteq \PageSet} \sum_{t \geq 0} 
\Td{S,t} \cdot
g_i(\SatProb{i}(S,t,\PermDist{i}); S,t). 
\end{align*}

Engine $i$ plays the non-deterministic strategy \PermDist{i}.
For each permutation \Perm, let $\PermDist{i}(\Perm)$ be the
probability that engine $i$ chooses the permutation \Perm. 
Recall from Section~\ref{sec:game-setup}
that
\begin{align}
\SatProb{i}(S,t,\PermDist{i}) 
& = \sum_{\Perm} \PermDist{i}(\Perm) \cdot
\Indicator{\PermS{t} \cap S \neq \emptyset};
\label{eqn:convex-satisfaction-prob}
\end{align}
thus, $\SatProb{i}(S,t,\PermDist{i})$ is a convex combination of 
the $\Indicator{\PermS{t} \cap S \neq \emptyset}$ terms.
In the first two cases of the theorem, we can use the 
convexity of $g_i$ to obtain that
\begin{align}
g_i(\SatProb{i}(S,t,\PermDist{i}); S,t)
& \leq \sum_{\Perm} \PermDist{i}(\Perm) \cdot
g_i(\Indicator{\PermS{t} \cap S \neq \emptyset}; S,t).
\label{eqn:convex-bound}
\end{align}
In the third case ($\DF = 0$, no assumptions on \SelRule), 
Inequality~\eqref{eqn:convex-bound} holds for a slightly different
reason: we can write
\begin{align*}
g_i(\SatProb{i}(S,t,\PermDist{i}); S,t)
& = \int_{\PermDists_{-i}} 
\SelRule[i]((\SatProb{i}(S,t,\PermDist{i}),
             \SatProbs_{-i}(S,t,\PermDists_{-i}))) 
\cdot \SatProb{i}(S,t,\PermDist{i}) \cdot
\JointDensity[-i]{\PermDists_{-i}} d \PermDists_{-i}\\
& \leq \int_{\PermDists_{-i}} 
\SelRule[i]((1, \SatProbs_{-i}(S,t,\PermDists_{-i}))) 
\cdot \sum_{\Perm} \PermDist{i}(\Perm) \cdot
                   \Indicator{\PermS{t} \cap S \neq \emptyset}
\cdot \JointDensity[-i]{\PermDists_{-i}} d \PermDists_{-i}\\
& = \sum_{\Perm} \PermDist{i}(\Perm) \cdot 
\int_{\PermDists_{-i}} 
\SelRule[i]((\Indicator{\PermS{t} \cap S \neq \emptyset}, 
             \SatProbs_{-i}(S,t,\PermDists_{-i}))) 
\cdot \Indicator{\PermS{t} \cap S \neq \emptyset}
\cdot \JointDensity[-i]{\PermDists_{-i}} d \PermDists_{-i}\\
& = \sum_{\Perm} \PermDist{i}(\Perm) \cdot
g_i(\Indicator{\PermS{t} \cap S \neq \emptyset}; S,t).
\end{align*}
For the inequality, we used monotonicity of the selection rule.
In the penultimate step, we used that whenever 
$\Indicator{\PermS{t} \cap S \neq \emptyset} \neq 1$, then
$\Indicator{\PermS{t} \cap S \neq \emptyset} = 0$, so the whole
expression for the particular \Perm is 0 anyway.

Next, we argue that Inequality~\eqref{eqn:convex-bound} is in fact
strict for at least one pair $(S,t)$ with $\Td{S,t} > 0$, in all
cases.
Because \PermDist{i} is not a deterministic strategy,
there exist two non-equivalent permutations 
$\Perm, \Perm'$ with $\PermDist{i}(\Perm) \in (0,1)$ and
$\PermDist{i}(\Perm') \in (0,1)$.
Let $X \subseteq \TypeSupp$ be the set of types 
satisfied by $\Perm$ but not by $\Perm'$,
and $X' \subseteq \TypeSupp$ the set of types 
satisfied by $\Perm'$ but not by $\Perm$; 
note that $X \cup X' \neq \emptyset$ by non-equivalence.
We then have $0 < \SatProb{i}(S,t,\MatchDist{i}) < 1$ 
for all $(S,t) \in X \cup X'$.
Therefore, in Equation~\eqref{eqn:convex-satisfaction-prob} for
those types $(S,t)$, there are both terms with
$\Indicator{\PermS{t} \cap S \neq \emptyset} = 0$ and with
$\Indicator{\PermS{t} \cap S \neq \emptyset} = 1$.
Thus, for the first two cases (when $g_i(\cdot; S, t)$ is strictly
convex), Inequality~\eqref{eqn:convex-bound} is strict for this
$(S,t)$ pair.\footnote{It would not have been strict if all indicator
  terms had been equal.}

Again, the reason for the strictness is slightly different in the
third case. Focus on the searcher types $(S,t) \in X \cup X'$.
Assume for contradiction that $\SatProb{j}(S,t,\MatchDist{j}) = 0$ for
all $j \neq i$ and $(S,t) \in X \cup X'$, and that
Inequality~\eqref{eqn:convex-bound} is not strict.
Thus, for each $(S,t), (S',t') \in X \cup X'$,
\begin{align*} 
\SelRule[i](\SatProbs(S,t,\PermDists)) 
& = \SelRule[i](1, \vc{0}) 
\; = \; \SelRule[i](\SatProbs(S',t',\PermDists)).
\end{align*}

Because $\Gamma$ is in general position, we know
$\sum_{(S,t) \in X} \Td{S,t} \neq \sum_{(S,t) \in X'} \Td{S,t}$;
w.l.o.g., assume that the left-hand side is larger.
Then, if engine $i$ were to increase the probability of the
permutation \Perm by some $\epsilon > 0$ while reducing the
probability of $\Perm'$ by $\epsilon$, 
her resulting change in utility would be 
\begin{multline*} 
\epsilon \cdot \Big( 
\sum_{(S,t) \in X} \Td{S,t} \SelRule[i](\SatProbs(S,t,\PermDists)) 
- \sum_{(S,t) \in X'} \Td{S,t} \SelRule[i](\SatProbs(S,t,\PermDists))
\Big) \\
 = \epsilon \cdot \SelRule[i](1, \vc{0}) \cdot
\Big( \sum_{(S,t) \in X} \Td{S,t} - \sum_{(S,t) \in X'} \Td{S,t} \Big)
\; > \; 0.
\end{multline*}
Engine $i$ could thus increase her utility by shifting
probability from $\Perm'$ to $\Perm$,
contradicting the equilibrium assumption.
We conclude that there must be 
another engine $j$ with $\SatProb{j}(S,t,\MatchDist{j}) > 0$.
Now, the non-indifference of the selection rules implies that for
at least one such engine $j$ (possibly $j=i$), 
$\SelRule[j](\SatProbs) < \SelRule[j](1, \SatProbs_{-j})$;
for this particular engine $j$, this implies
strictness of Inequality~\eqref{eqn:convex-bound}.
The remainder of the proof now proceeds with this engine (formerly
$j$) as engine $i$.

Finally, summing up the strict Inequality~\eqref{eqn:convex-bound}
over all types $(S,t) \in \TypeSupp$
and changing orders of summation, we obtain that 
\begin{align}
\Payoff{i}{\PermDists}
& < 
\sum_{\Perm} \PermDist{i}(\Perm) \cdot
\sum_{S \subseteq \PageSet} \sum_{t \geq 0} 
g_i(\Indicator{\PermS{t} \cap S \neq \emptyset}; S,t).
\label{eqn:strict-utility-inequality}
\end{align}

The right-hand side of
Inequality~\eqref{eqn:strict-utility-inequality} is exactly the
utility of engine $i$ when \emph{mixing} over deterministic strategies
instead of playing a non-deterministic strategy. 
Thus, \PermDist{i} cannot have been engine $i$'s best strategy at the
correlated equilibrium under signal \SIGNAL[i], a contradiction.
\end{proof}


\subsection{An Upper Bound on the Correlated Price of Anarchy}
\label{sec:convex-PoA-bound}

We now show that if engines apply deterministic strategies at equilibrium,
the resulting outcomes will generate high social welfare.  Combining this result
with Theorem \ref{thm:convex.deterministic}, we conclude that the PoA
of \SearchGame{\DF}{\TypeDist} is at most $2$ under the conditions
of Theorem \ref{thm:convex.deterministic}.
In fact, we prove a slightly better bound of
$\frac{2\NUMENG-1-\DF}{\NUMENG-\beta}$ on the PoA, and show that this 
bound is tight for all $\DF \in [0,1]$ and $\NUMENG \geq 2$.

\begin{theorem} \label{thm:bounded-PoA-for-deterministic}
Suppose that \SelRule is non-indifferent
and every correlated equilibrium of \SearchGame{\DF}{\TypeDist} is
deterministic. 
Then, the correlated PoA of \SearchGame{\DF}{\TypeDist}
is at most $\frac{2\NUMENG-1-\DF}{\NUMENG-\beta}$. 
\end{theorem}

\begin{proof}
Consider a correlated deterministic equilibrium of \SearchGame{\DF}{\TypeDist},
characterized by a joint distribution \JointDensity{\Perms}
over vectors of permutations \Perms.
Let $\OptPerms$ be the vector of permutations for each engine under
the socially optimal profile.
Focus on one engine $i$ and one signal \SIGNAL[i] she receives. 
Let \Perm[i] be the permutation played by $i$ when seeing \SIGNAL[i].
From the perspective of $i$, conditioned on receiving
\SIGNAL[i], the other engines are playing a correlated distribution 
$\JointDensity[-i]{\Perms_{-i}}$ over deterministic strategies.

Let
$\mu_{S,t}^{\SIGNAL[i]} = \ProbC{\forall i': \PermS[i']{t} \cap S = \emptyset}{\SIGNAL[i]}$
be the probability that no engine $i'$ (including $i$) satisfies the
searcher type $(S,t)$, given that $i$ received the signal \SIGNAL[i].
Consider a deviation by $i$ to any $\OptPerm[j]$.
This will certainly attract any searcher who would not be satisfied by
any other engine, but would be satisfied by $\OptPerm[j]$. 
Also, any other searcher satisfied by $\OptPerm[j]$ would select
engine $i$ with probability at least $\SelRule[i](\vc{1})$,
by the assumed monotonicity of the selection rule.  
Finally, this deviation will
not affect the behavior of searchers that are not satisfied by any engine
or $\OptPerm[j]$.
Putting this together,
the expected conditional utility of this deviation is at least
\begin{align*}
\mu_j 
& = \sum_{S \subseteq \PageSet} \sum_{t \geq 0} 
\Td{S,t} \cdot \Indicator{\OptPermS[j]{t} \cap S \neq \emptyset}
\cdot \mu_{S,t}^{\SIGNAL[i]} \\
& + \sum_{S \subseteq \PageSet} \sum_{t \geq 0}
\Td{S,t} \cdot \SelRule[i](\vc{1}) \cdot \Indicator{\OptPermS[j]{t} \cap S \neq \emptyset}
\cdot (1-\mu_{S,t}^{\SIGNAL[i]}) \\
& +
\beta \cdot \sum_{S \subseteq \PageSet} \sum_{t \geq 0} 
\Td{S,t} \cdot \SelRule[i](\vc{0}) \cdot \Indicator{\OptPermS[j]{t} \cap S = \emptyset}
\cdot \mu_{S,t}^{\SIGNAL[i]}.
\end{align*}

Since engine $i$, receiving the signal \SIGNAL[i], does not want to
deviate to \OptPerm[j] at the correlated equilibrium, we have

\begin{align*}
\mu_j & \leq 
\int_{\Perms_{-i}}
\sum_{S \subseteq \PageSet} \sum_{t \geq 0} 
\Td{S,t} \cdot \SelRule[i](\SatProbs(S,t,\Perms)) 
\cdot \SatProb{i}(S,t,\Perms)
\cdot \JointDensity[-i]{\Perms_{-i}} d\Perms_{-i} 
+ \DF \cdot \sum_{S \subseteq \PageSet} \sum_{t \geq 0} 
\Td{S,t} \cdot \SelRule[i](\vc{0})
\cdot \mu_{S,t}^{\SIGNAL[i]}.
\end{align*}
The first term is the utility that $i$ obtains from
searchers she attracts and satisfies, while the second term is the
utility obtained from searchers that are not satisfied
by any engine, and visit engine $i$ (failing to be satisfied).
Taking an expectation over signals $Z_i$,
writing $\mu_{S,t} = \Expect[Z_i \sim \MarginalDensity{i}]{\mu_{S,t}^{Z_i}}$, 
and rearranging the terms with \DF factors, we conclude
\begin{align*}
& \sum_{S \subseteq \PageSet} \sum_{t \geq 0} 
\Td{S,t} \cdot \Indicator{\OptPermS[j]{t} \cap S \neq \emptyset}
\cdot \mu_{S,t} 
+ \sum_{S \subseteq \PageSet} \sum_{t \geq 0}
\Td{S,t} \cdot \SelRule[i](\vc{1}) \cdot \Indicator{\OptPermS[j]{t} \cap S \neq \emptyset}
\cdot (1-\mu_{S,t}) \\
& \leq 
\int_{\Perms}
\sum_{S \subseteq \PageSet} \sum_{t \geq 0} 
\Td{S,t} \cdot \SelRule[i](\SatProbs(S,t,\Perms)) 
\cdot \SatProb{i}(S,t,\Perms)
\cdot \JointDensity{\Perms} d\Perms 
\\ & \qquad 
+ \DF \cdot \sum_{S \subseteq \PageSet} \sum_{t \geq 0} 
\Td{S,t} \cdot \SelRule[i](\vc{0}) \cdot \Indicator{\OptPermS[j]{t} \cap S \neq \emptyset}
\cdot \mu_{S,t}.
\end{align*}
Taking a sum over all $i$, and applying the definition of equilibrium
welfare $\Welfare{\JOINTDENSITY}$, we get
\begin{align*}
& \NUMENG \cdot \sum_{S \subseteq \PageSet} \sum_{t \geq 0} 
\Td{S,t} \cdot \Indicator{\OptPermS[j]{t} \cap S \neq \emptyset}
\cdot \mu_{S,t} 
+ \sum_{S \subseteq \PageSet} \sum_{t \geq 0} 
\Td{S,t} \cdot \Indicator{\OptPermS[j]{t} \cap S \neq \emptyset}
\cdot (1-\mu_{S,t}) \cdot \sum_i \SelRule[i](\vc{1}) \\
& \leq 
\Welfare{\JOINTDENSITY}
+ \DF \cdot \sum_{S \subseteq \PageSet} \sum_{t \geq 0} 
\Td{S,t} \cdot \mu_{S,t} \cdot \Indicator{\OptPermS[j]{t} \cap S \neq \emptyset} 
\cdot \sum_i \SelRule[i](\vc{0}).
\end{align*}
Using $\sum_i \SelRule[i](\vc{0}) \leq 1$, $\sum_i \SelRule[i](\vc{1}) = 1$,
and summing over all $j$,
\begin{align*}
\NUMENG \cdot \Welfare{\JOINTDENSITY} 
& \geq 
(\NUMENG - \DF) \cdot \sum_{S \subseteq \PageSet} \sum_{t \geq 0} 
\Td{S,t} 
\mu_{S,t} \sum_j \Indicator{\OptPermS[j]{t} \cap S \neq \emptyset}
\\ & \qquad 
+ \sum_{S \subseteq \PageSet} \sum_{t \geq 0} 
\Td{S,t} (1-\mu_{S,t}) \sum_j \Indicator{\OptPermS[j]{t} \cap S \neq \emptyset} \\
& \geq 
(\NUMENG - \DF) \cdot \sum_{S \subseteq \PageSet} \sum_{t \geq 0} 
\Td{S,t} \mu_{S,t} \cdot \Indicator{\exists j \colon \OptPermS[j]{t}
  \cap S \neq \emptyset} 
\\ & \qquad 
+ \sum_{S \subseteq \PageSet} \sum_{t \geq 0} 
\Td{S,t} (1-\mu_{S,t}) \cdot \Indicator{\exists j \colon \OptPermS[j]{t} \cap S \neq \emptyset}  \\
& = (\NUMENG - \DF)\left( \Welfare{\OPT} - \lambda \cdot \Welfare{\JOINTDENSITY} \right) + \lambda \cdot \Welfare{\JOINTDENSITY}
\end{align*}
for some $\lambda \in [0,1]$.  
The last equality follows because
\begin{align*}
\sum_{S \subseteq \PageSet} \sum_{t \geq 0} 
\Td{S,t} \cdot \Indicator{\exists j \colon \OptPermS[j]{t} \cap S \neq \emptyset}
& = \Welfare{\OPT} \quad \mbox{ and}\\
\sum_{S \subseteq \PageSet} \sum_{t \geq 0} \Td{S,t} (1-\mu_{S,t}) 
& \leq \Welfare{\JOINTDENSITY}.
\end{align*}
Rearranging, we have 
$\Welfare{\JOINTDENSITY} 
\geq \frac{\NUMENG - \DF}{(1+\lambda)\NUMENG - \lambda(\DF+1)} \cdot \Welfare{\OPT}$, 
with $\lambda \in [0,1]$.
As $\NUMENG \geq 2$ and $\DF \in [0,1]$, this 
expression is minimized at $\lambda = 1$; hence,
$\Welfare{\JOINTDENSITY} \geq \frac{\NUMENG - \DF}{2\NUMENG - \DF-1}\Welfare{\OPT}$,
as claimed.
\end{proof}

We now show that the PoA bound from Theorem
\ref{thm:bounded-PoA-for-deterministic} is tight for all 
$\NUMENG \geq 2$ and $\DF \in [0,1]$, 
even for the pure PoA of the \SingletonGameName with a convex selection rule.

\begin{proposition}
\label{prop:poa.n}
Consider the \SingletonGameName with any $\DF \in [0,1]$ and any
symmetric, strictly convex selection rule.  
There are instances in which the 
PoA is at least $\frac{2\NUMENG-1-\DF}{\NUMENG-\DF}$.
\end{proposition}
\begin{proof}
In the tight instance, there are \NUMENG pages,
and $\Td{1} = \frac{\NUMENG-\DF}{2\NUMENG-1-\DF}$, while
$\Td{n} = \frac{1}{2\NUMENG-1-\DF}$ for all $n > 1$.
Consider the profile in which each engine displays page 1.
For any symmetric selection rule (convex or otherwise), each engine
has utility 
\begin{align*}
\frac{\Td{1}}{\NUMENG} + \DF \cdot \sum_{n > 1}\frac{\Td{n}}{\NUMENG} 
& = \frac{\DF(\NUMENG-2)+\NUMENG}{\NUMENG(2\NUMENG-\DF-1)}.
\end{align*}
Since the selection rule is strictly convex, the proof of Theorem
\ref{thm:convex.deterministic} shows that the utility-maximizing
strategy of each engine is a deterministic strategy. 
Since the selection rule is symmetric and non-indifferent, the utility
of switching to another page $n' > 1$ is equal to
\begin{align*}
\Td{n'} + \DF \cdot \sum_{\substack{n>1\\n \neq n'}}\frac{\Td{n}}{\NUMENG} 
& = \frac{\DF(\NUMENG-2)+\NUMENG}{\NUMENG(2\NUMENG-\DF-1)}.
\end{align*}
Hence, the strategy profile is an equilibrium.
(In fact, by increasing \Td{1} by a very small amount, we can
  ensure that this strategy is the unique equilibrium.)
It satisfies a $\frac{\NUMENG-\DF}{2\NUMENG-1-\DF}$
fraction of users.
On the other hand, the profile in which engine $i$ displays page $i$
satisfies all users. 
The PoA in this example is therefore $\frac{\NUMENG-\DF}{2\NUMENG-1-\DF}$.
\end{proof}


\section{Symmetric Equilibria and Low Welfare}
\label{sec:sym}

Next, we analyze equilibria of the \SearchGameName when
the conditions of Theorem~\ref{thm:convex.deterministic} do not apply;
that is, when $\DF > 0$ and the selection rule is not convex.
We will show that the resulting misalignment of incentives can lead to equilibria
with very poor social welfare, i.e., high PoA.
All of our results are proved already for the
\SingletonGameName\footnote{The strong lower bounds continue
    to hold under the (realistic) assumption that $t \ll \NUMPAGES$. 
    When $t \approx \NUMPAGES$, searchers will typically see most
    pages, so good welfare is obtained.}.

\subsection{When Satisfaction is Irrelevant: $\DF = 1$}
\label{sec:beta1}

We begin by considering the case $\DF = 1$.
Our first result is that when the selection rule \SelRule is symmetric
and satisfies a cross-concavity assumption, every pure Nash
Equilibrium of the \SingletonGameName is \emph{symmetric}.
In other words, search engines do not specialize: all engines apply
the same strategy. A user can thus do no better than selecting an
engine arbitrarily, and hence competition does not positively
impact social welfare.

\subsection{Symmetric Equilibria}
\label{sec:symmetric}
We show that, under certain cross-concavity conditions on the
selection rule \SelRule, all pure Nash Equilibria of the
\SingletonGameName are symmetric.
We say that \SelRule is \emph{symmetric} if the value of
$\SelRule[i](\SatProbs)$ is invariant under relabeling of the engines,
and write $\SelRule(\SatProbs) = \SelRule[1](\SatProbs)$
for symmetric selection rules, 
as this fully specifies user behavior.
\SelRule is \emph{strictly cross-concave} if it satisfies the following
property: for all $x$ and $y$ such that $x < y$, 
there exists a $\delta > 0$ such that for 
all $0 < \epsilon < \delta$,
\begin{align*}
\SelRule(x+\epsilon,y,\SatProbs_{-\SET{1,2}}) 
- \SelRule(x,y,\SatProbs_{-\SET{1,2}})
& > 
\SelRule(y,x,\SatProbs_{-\SET{1,2}})
- \SelRule(y-\epsilon,x,\SatProbs_{-\SET{1,2}}).
\end{align*}
Intuitively, if \SelRule is cross-concave, 
the increase in user visits that a (trailing) engine experiences when
increasing her satisfaction probability is at least as large as
the loss in visits incurred by a (leading) engine from decreasing her
satisfaction probability by the same amount.
That is, catching up in user visits is easier than expanding a lead.


\begin{theorem}
\label{thm:beta1.sym}
Let \SatMat be a pure strategy equilibrium of the
\SingletonGameName with $\DF = 1$ and symmetric and strictly
cross-concave selection function \SelRule.
Then \SatMat must be symmetric.
\end{theorem}

\begin{proof}
The proof proceeds by directly arguing
that if two engines employ different strategies, then cross-concavity implies
that one of them can strictly improve her welfare by shifting her page
distribution slightly toward the strategy used by the other.

Let $\SatMat = (\SatProbV{i})_{i=1}^{\NUMENG}$ be a pure strategy
equilibrium of the \SingletonGameName.
Assume for contradiction that \SatMat is not symmetric.
Then, there are two engines 
(without loss of generality, engines $1$ and $2$) 
and two pages (without loss of generality, pages $1$ and $2$)
such that $\SP{1}{2} < \SP{2}{2}$ and $\SP{1}{1} > \SP{2}{1}$.
Let $\delta_1$ be the constant in the definition of cross-concavity
when applied to engines $1,2$ with $x = \SP{1}{2}$ and $y = \SP{2}{2}$.
Similarly, let $\delta_2$ be the constant when the definition is
applied to engines $1,2$ with $x = \SP{2}{1}$ and $y = \SP{1}{1}$. 
Fix an arbitrary $\epsilon < \min(\delta_1, \delta_2)$.
By strict cross-concavity,
\begin{equation} \label{eqn:cc-impl} \begin{split}
& \SelRule(\SP{1}{2}+\epsilon,\SP{2}{2},\SatProbs_{-\SET{1,2}}) 
- \SelRule(\SP{1}{2},\SP{2}{2},\SatProbs_{-\SET{1,2}}) 
\\ & \; > 
\SelRule(\SP{2}{2},\SP{1}{2},\SatProbs_{-\SET{1,2}})
- \SelRule(\SP{2}{2}-\epsilon,\SP{1}{2},\SatProbs_{-\SET{1,2}}),\\
& \SelRule(\SP{2}{1}+\epsilon,\SP{1}{1},\SatProbs_{-\SET{1,2}}) 
- \SelRule(\SP{2}{1},\SP{1}{1},\SatProbs_{-\SET{1,2}}) 
\\ & \; > 
\SelRule(\SP{1}{1},\SP{2}{1},\SatProbs_{-\SET{1,2}})
- \SelRule(\SP{1}{1}-\epsilon,\SP{2}{1},\SatProbs_{-\SET{1,2}}).
\end{split} \end{equation}

Because \SatMat is an equilibrium, engine 1 cannot profit from moving
$\epsilon$ probability from page 1 to page 2, nor can engine 2 profit
from moving $\epsilon$ from page 2 to page 1. Thus,
\begin{equation} \label{eqn:ne-impl} \begin{split}
& \SelRule(\SP{1}{2}+\epsilon,\SP{2}{2},\SatProbs_{-\SET{1,2}}) 
+ \SelRule(\SP{1}{1}-\epsilon,\SP{2}{1},\SatProbs_{-\SET{1,2}}) \\
& \hspace{0.5cm} \leq
\SelRule(\SP{1}{2},\SP{2}{2},\SatProbs_{-\SET{1,2}}) 
+ \SelRule(\SP{1}{1},\SP{2}{1},\SatProbs_{-\SET{1,2}}),\\
& \SelRule(\SP{2}{1}+\epsilon,\SP{1}{1},\SatProbs_{-\SET{1,2}}) 
+ \SelRule(\SP{2}{2}-\epsilon,\SP{1}{2},\SatProbs_{-\SET{1,2}}) \\
& \hspace{0.5cm} \leq
\SelRule(\SP{2}{1},\SP{1}{1},\SatProbs_{-\SET{1,2}}) 
+ \SelRule(\SP{2}{2},\SP{1}{2},\SatProbs_{-\SET{1,2}}).
\end{split} \end{equation}

Adding the two Equations~\eqref{eqn:cc-impl} and subtracting the two
Equations~\eqref{eqn:ne-impl}, all terms cancel and we are left
with the contradiction $0 > 0$.
Hence, \SatMat must be symmetric.
\end{proof}

Theorem~\ref{thm:beta1.sym} implies that for
symmetric and strictly cross-concave selection rules and $\DF=1$,
more engines do not provide any advantage in pure strategy
equilibria. As we show below, this implies a PoA of $\Omega(\NUMENG)$.
However, we will show that the effects of competition can be even more
negative, by giving lower bounds on the \emph{Price of Stability} for
the proportional selection rule.
Thereto, we next characterize the (unique)
equilibrium for the proportional selection rule.
The proof follows from an analysis of the equilibrium
condition that no engine wishes to shift probability mass from any
page to any other.

\begin{theorem}
\label{thm:beta1.prop}
When $\DF = 1$ and \SelRule is the proportional sharing rule, 
the \SingletonGameName has a unique pure Nash Equilibrium;
in this equilibrium,
each engine selects page $n \in \Omega$
with probability \Td{n}.
\end{theorem}

\begin{proof}
The proportional sharing rule satisfies symmetry and
strict cross-concavity.
Thus, Theorem~\ref{thm:beta1.sym} guarantees that all pure equilibria
are symmetric. It remains to show that a symmetric equilibrium \emph{exists} and
is \emph{unique}.

We begin by showing uniqueness and deriving a specific form for the
equilibrium. Let \SatMat be a symmetric pure equilibrium; it is
completely described by the (common) probability \SPr{n} with which
the engines choose page $n$.
We claim that $\SPr{n} > 0$ for all pages $n$.
Assume for contradiction that $\SPr{n} = 0$ for some page $n$, and let
$n'$ be any page with $\SPr{n'} > 0$.
Consider the deviation in which engine 1 adds a small amount
$\epsilon$ of probability to page $n$, and removes it from page $n'$.
Engine 1 receives all visits from users of type $n$, for an
added payoff of $\Td{n} > 0$ (since each page is
desired with positive probability).
The loss in payoff from page $n'$ is less than
$\Td{n'} \cdot \frac{\epsilon}{\NUMENG \SPr{n'}}$
For small enough $\epsilon$, this is strictly less than the gain
\Td{n}. Thus, \SatMat was not an equilibrium.

Because $\SPr{n} > 0$ for all $n$, the strategy profiles under
consideration are bounded away in all coordinates from $0$ (and 
also $1$), and thus from the discontinuities (and
non-differentiabilities) of the selection function. 
We can therefore safely take derivatives.

Engine $i$ is choosing her strategy to maximize
$\sum_n \Td{n} \SelRule(\SatProbPage{n})$ subject to 
$\sum_n \SP{i}{n} = 1$; thus, the derivative of the Lagrangian
$L(\SatProbV{i},\lambda) = 
\sum_n \Td{n} \SelRule(\SatProbPage{n})) + \lambda (1 - \sum_n \SP{i}{n})$
with respect to all variables must be 0:
\begin{align*}
\Td{n} \frac{d}{d \SP{i}{n}} \SelRule(\SP{i}{n}, \SatProbV{-i}(n)) 
& = \lambda
\end{align*}
for each page $n$.  
Using the definition of the proportional selection rule, this
implies
\begin{align}
\frac{\sum_{i' \neq i} \SP{i'}{n}}{(\sum_{i'} \SP{i'}{n})^2} 
& = \frac{\lambda}{\Td{n}}
\label{eqn:proportional-eq}
\end{align}
for each $n$.
As the equilibrium is symmetric, $\SP{i'}{n} = \SP{i}{n}$ for all
$i'$, simplifying Equation~\eqref{eqn:proportional-eq} to
$\SP{i}{n} = \frac{\Td{n} \cdot (\NUMENG-1)}{\lambda \NUMENG^2}.$
Since \SatProbV{i} and \TypeDist are distributions, we must have
$\lambda \NUMENG^2/(\NUMENG-1) = 1$, 
and hence $\SP{i}{n} = \Td{n}$ for all engines
$i$ is the only equilibrium.

It remains to prove that this distribution is in fact an equilibrium.
Consider any best response \SatProbV{i} by engine $i$ to the
symmetric equilibrium $\SP{i'}{n} = \Td{n}$ played by the other
engines. 
We use variants of the Lagrangian condition
\eqref{eqn:proportional-eq}, substituting that $\SP{i'}{n} = \Td{n}$.
Assume that $\SatProbV{i} \neq \SatProbV{i'}$; thus, there must exist
pages $n,n'$ with $0 \leq \SP{i}{n} < \SP{i'}{n}$
and $\SP{i'}{n'} < \SP{i}{n'} \leq 1$.
Since it is possible that $\SP{i}{n} = 0$ or $\SP{i}{n'} = 1$, we only
obtain weaker Lagrangian conditions of non-positive (resp.,
non-negative) derivative, as opposed to derivatives of 0.
These weaker Lagrangian conditions for $n,n'$ yield
\begin{align*}
\frac{(\NUMENG-1) \Td{n}}{(\SP{i}{n} + (\NUMENG-1) \Td{n})^2} 
& \leq \frac{\lambda}{\Td{n}},
& \frac{(\NUMENG-1) \Td{n'}}{(\SP{i}{n'} + (\NUMENG-1) \Td{n'})^2} 
& \geq \frac{\lambda}{\Td{n'}}.
\end{align*}
We can rearrange for $\lambda$ and use that 
$\SP{i}{n} < \Td{n}$ and $\Td{n'} < \SP{i}{n'}$ to obtain that
\begin{align*}
& \frac{(\NUMENG-1) \Td{n}^2}{(\NUMENG \Td{n})^2}
\; < \; \frac{(\NUMENG-1) \Td{n}^2}{(\SP{i}{n} + (\NUMENG-1) \Td{n})^2} 
\; \leq \; \lambda
\; \leq \; \frac{(\NUMENG-1) \Td{n'}^2}{(\SP{i}{n} + (\NUMENG-1) \Td{n'})^2} 
\; < \; \frac{(\NUMENG-1) \Td{n'}^2}{(\NUMENG \Td{n'})^2},
\end{align*}
a clear contradiction after canceling out the \Td{n} terms on the left
and the \Td{n'} terms on the right. 
Thus, the only locally optimal response for player $i$ is to play the
same distribution as the other players, so the symmetric distribution
is an equilibrium.
\end{proof}

In summary, when users apply the proportional selection rule, it is a
(unique) symmetric equilibrium for each engine to choose its page to display
according to the type distribution of the users.

\subsection{Lower Bounds on the Price of Stability}

We use Theorem \ref{thm:beta1.prop} to bound the
PoS for the \SingletonGameName (and thus also the
\SearchGameName).

\begin{proposition}
\label{prop:poa.n}
Consider the \SingletonGameName with $\DF=1$ and the proportional
selection rule.
\begin{enumerate}
\vspace{-1mm}
\item There are instances in which the PoS is at least
  \NUMENG.
\vspace{-2mm}
\item There are instances with only two engines in which the PoS 
is $\Omega(\sqrt{\NUMPAGES})$.
\end{enumerate}
\end{proposition}

\begin{proof}
Suppose that we have $\NUMPAGES \geq \NUMENG$ pages, and assume
w.l.o.g.~that $\Td{1} > \Td{2} > \cdots > \Td{\NUMPAGES}$, and the
probabilities are in general position.
The social optimum occurs when engine $i$ deterministically displays
page $i$, leading to a social welfare of 
$\sum_{n=1}^{\NUMENG} \Td{n}$.
By Theorem \ref{thm:beta1.prop}, the unique pure equilibrium has 
each engine choosing each page $n$ with probability \Td{n}.
Each user will then be satisfied with probability 
$\sum_{n=1}^{\NUMPAGES} \Td{n}^2$, regardless of his engine choice.
The PoS is thus 
$\frac{\sum_{n=1}^{\NUMENG} \Td{n}}{\sum_{n=1}^{\NUMPAGES} \Td{n}^2}$.

We now consider two special cases.
First, when the first \NUMENG pages have probability $\Theta(1/\NUMENG)$
each, and the remaining pages together have probability $o(1/\NUMENG)$, the
PoS is $\Omega(\NUMENG)$, as can be seen by
upper-bounding $\sum_{n=1}^{\NUMPAGES} \Td{n}^2 
\leq \sum_{n=1}^{\NUMENG} \Td{n}^2 + (\sum_{n > \NUMENG} \Td{n})^2$.
Another distribution giving rise to this type of behavior is when all
pages have probability $\Theta(1/\NUMPAGES)$.

Second, consider the following distribution:
the top two pages have probability $\Theta(1/\sqrt{\NUMPAGES})$;
the remaining pages have probability $\Theta(1/\NUMPAGES)$.
The PoS becomes 
$$\Theta(1/\sqrt{\NUMPAGES})/\Theta(1/\NUMPAGES) = \Theta(\sqrt{\NUMPAGES}).$$
Thus, when there are a few highly relevant pages and many rare
pages, the symmetric equilibrium leads to high
inefficiency $\Omega(\sqrt{\NUMPAGES})$.
\end{proof}

\section{When Satisfaction is Partially Beneficial: $\DF \in (0,1)$}
\label{sec:intermediate-beta}

The previous two sections describe equilibria of two very different
forms: an equilibrium with deterministic and potentially non-symmetric
page selections for the case $\DF = 0$, and a symmetric equilibrium
that follows the user type distribution for the case $\DF = 1$ under
the proportional selection rule.
We now consider the intermediate case, $\DF \in (0,1)$, which
represents scenarios in which search engines derive some, but
not full, expected advertising benefit from a user who is not served a
desired page.

We show that, for any fixed $\DF$, our lower bounds on the PoS
in Proposition~\ref{prop:poa.n} can persist
as long as the number of pages \NUMPAGES is large.
(However, for intermediate \DF values, we cannot rule out the
existence of other equilibria, so we obtain lower
bounds on the PoA, not the PoS.)
Our conclusion is that the constant PoA bounds at 
$\DF = 0$ in general do not survive perturbations of $\DF$
when the selection rule is not convex.

\ignore{
When $\DF > 0$, but $\NUMPAGES$ is not sufficiently large (or,
equivalently, $\DF$ is too small for a given fixed $\NUMPAGES$), it
may be that an equilibrium does not exist at all.  In this case, we
can use the continuity of the utility functions as functions of \DF to
argue that a deterministic strategy profile forms a \DF-approximate
equilibrium. 
}

\subsection{Existence and Uniqueness of Symmetric Equilibria}
\label{sec:symmetric-intermediate-beta}

Since our goal in this section --- much like in
Section~\ref{sec:beta1} --- is to establish strong lower bounds on the
Price of Anarchy,
we again focus on the \SingletonGameName with the
proportional selection rule \SelRule.
We first show that the existence of symmetric 
equilibria for this class, established for $\DF = 1$ 
in Section~\ref{sec:beta1}, continues to hold under small
perturbations of $\DF$.
The proof is conceptually similar to that of
Theorem~\ref{thm:beta1.prop}.

\begin{theorem}
\label{thm:intermediate.sym}
Suppose that $\DF \in (0,1)$.
The \SingletonGameName with proportional selection rule \SelRule
has at most one symmetric equilibrium.
Moreover, whenever $\DF > 1 - 1/\NUMENG$, a symmetric equilibrium exists.
\end{theorem}

\begin{proof}
As in the proof of Theorem~\ref{thm:beta1.prop},
we can identify strategy profiles with 
satisfaction probabilities \SatMat, where $\SP{i}{n}$ is the
probability that engine $i$ selects page $n \in \Omega$. 
Suppose that \SatMat is a symmetric equilibrium, say with
$\SP{i}{n} = \SPr{n}$ for all $i$ and $n$.  Then,
as in the proof of Theorem~\ref{thm:beta1.prop},
$\SPr{n} > 0$ for all $n$;
otherwise, an engine could increase her utility by selecting page $n$
with sufficiently small probability.
We can therefore assume that $\SPr{n} > 0$ for all $n$. 

Because each \SatProbV{i} is a best response to the other engines'
strategies, it maximizes
\begin{align*}
\Payoff{i}{\SatProbV{i},\SatProbV{-i}}
& = \sum_n \Td{n} \SelRule(\SatProbPage{n}) \cdot (\DF + (1-\DF)\SP{i}{n}) \\
& = \; \sum_n \Td{n} \frac{\SP{i}{n}}{\sum_k \SP{k}{n}} 
\cdot (\DF + (1-\DF)\SP{i}{n}),
\end{align*}
subject to $\sum_n \SP{i}{n} = 1$.
The Lagrangian for this constrained optimization is
\begin{align*}
L_i (\lambda_i, \SatProbV{i})
& = 
\sum_n \Td{n} \frac{\SP{i}{n}}{\SP{i}{n} + (\NUMENG-1) \SPr{n}}
      \cdot (\DF + (1-\DF)\SP{i}{n}) 
+ \lambda_i(1 - \sum_n \SP{i}{n}).
\end{align*}
Because $\SP{i}{n} > 0$ for all $n$, we can take derivatives;
a necessary condition for \SatProbV{i} to be a local optimum for $L_i$
is then that $\frac{d}{d \SP{i}{n}} L_i(\lambda_i, \SatProbV{i}) = 0$ 
for each $n$. 
This yields
\begin{align*}
\Td{n} \frac{(\NUMENG-1) \SPr{n} \cdot (2(1-\DF)\SP{i}{n} + \DF) +
  (1-\DF)\SP{i}{n}^2}{(\SP{i}{n} + (\NUMENG-1) \SPr{n})^2} = \lambda_i.
\end{align*}
Because \SatMat is symmetric, we also have that $\SP{i}{n} = \SPr{n}$
for all $n$, and can simplify this condition to
\begin{align*}
\Td{n} \frac{(\NUMENG-1) \DF + (2\NUMENG-1)(1-\DF) \SPr{n}}{\NUMENG^2 \SPr{n}} 
& = \lambda_i.
\end{align*}
Writing 
$\lambda'_i = \frac{\NUMENG^2 \lambda_i}{(\NUMENG-1) \DF}$,
$z = \frac{(2\NUMENG-1)(1-\DF)}{(\NUMENG-1) \DF}$,
and rearranging terms, we can solve for \SPr{n} to get
\begin{align}
\SPr{n} 
& = \frac{\Td{n}}{\lambda'_i - z \Td{n}}.
\label{eqn:beta-symmetric-response}
\end{align}
Because $\SPr{n} > 0$, this implies
$\lambda'_i > z \Td{n}$ for all $n$.
The condition $\sum_n \SPr{n} = 1$ implies
\begin{align}
\sum_n \frac{\Td{n}}{\lambda'_i - z \Td{n}} = 1. 
\label{eqn:lambda-condition}
\end{align}
Now let $n^* \in \argmax_n \SET{\Td{n}}$ be a page maximizing the
probability of being desired.
For $\lambda$ in the range $(z \Td{n^*}, \infty)$, we have
$\frac{\Td{n}}{\lambda - z \Td{n}} > 0$ for all pages $n$,
and the expression $\sum_n \frac{\Td{n}}{\lambda - z \Td{n}}$ is
strictly decreasing, approaching $0$ as $\lambda \to \infty$ and
approaching $\infty$ as $\lambda \to z \Td{n^*}$.
There is therefore a unique solution (in $\lambda'_i$) to the
Equation~\eqref{eqn:lambda-condition} in the range $(z \Td{n^*}, \infty)$.
By choice of $n^*$, all terms of 
$\sum_n \frac{\Td{n}}{\lambda'_i - z \Td{n}}$ are positive at this
solution, and each $\frac{\Td{n}}{\lambda'_i - z \Td{n}}$ lies in $(0,1)$.
Thus, the $\frac{\Td{n}}{\lambda'_i - z \Td{n}}$
form a valid distribution over pages, and constitute the only local
maximum of the Laplacian.
Thus, the values $\SPr{n} = \frac{\Td{n}}{\lambda'_i - z \Td{n}}$
constitute the only candidate symmetric equilibrium.\footnote{We
remark that for $\beta=1$, we get that $z=0$, so we recover the
analysis from the proof of Theorem~\ref{thm:beta1.prop} as a special
case.}

It remains to show that this symmetric profile, which attains a local
maximum for the utility of each engine, is actually a global maximum
(and thus a best response) for each player when $\DF > 1 - 1/\NUMENG$.
To see this, note that
$\frac{d^2 \Payoff{i}{\SatProb{i},\SatProbV{-i}}}{d\SP{i}{n} d\SP{i}{n'}} = 0$
for all pairs of pages $n \neq n'$;
furthermore, for each page $n$,
\begin{align}
\frac{d^2 \Payoff{i}{\SatProb{i},\SatProbV{-i}}}{d^2 \SP{i}{n}} 
& = 
- \frac{2 \Td{n} (\sum_{k \neq i} \SP{k}{n}) \cdot
        (\DF - (1 - \DF)\sum_{k \neq i}\SP{k}{n})}{%
        (\sum_k \SP{k}{n})^3}. 
\label{eqn:hessian-diagonal}
\end{align}
Since $\sum_{k \neq i}\SP{k}{n} \leq \NUMENG-1$, our assumption 
that $\DF > 1 - \frac{1}{\NUMENG}$ implies that
$\DF - (1 - \DF) \sum_{k \neq i}\SP{k}{n}) > 0$, 
and hence 
$\frac{d^2 \Payoff{i}{\SatProb{i},\SatProbV{-i}}}{d^2 \SP{i}{n}} < 0$
over the entire domain of \Payoff{i}{\cdot}.  
This implies that the Hessian of \Payoff{i}{\cdot} is negative
definite, and \Payoff{i}{\cdot} is strictly concave.
Thus, the local maximum described above is also a global maximum.
We conclude that the symmetric profile is an equilibrium.
\end{proof}

\subsection{Inefficiency of Equilibria}
Having shown the existence of symmetric equilibria for sufficiently
large \DF, we now show that for specific instances, symmetric
equilibria (and thus high PoA) arise for much smaller
values of \DF already.
Specifically, we show that the bad examples from
Proposition~\ref{prop:poa.n} persist for 
$\DF = \Omega(\NUMENG/\NUMPAGES)$ 
and $\DF = \Omega(\NUMENG/\sqrt{\NUMPAGES})$, respectively.
Notice that since $\NUMPAGES \gg \NUMENG$, this means that in
realistic settings, we cannot rule out high inefficiency due to search
engine competition for any constant value of $\DF$.

We emphasize that while  
Proposition~\ref{prop:poa.n} provided lower bounds on the \emph{Price
  of Stability}, we have not ruled out the existence of asymmetric
equilibria for $\DF \in (0,1)$; therefore, the propositions in this
section only prove lower bounds on the \emph{Price of Anarchy}.

\begin{proposition}
\label{prop:intermediate-engines}
Consider the \SingletonGameName with proportional selection rule
\SelRule. 
If $\DF > \frac{2\NUMENG}{\NUMPAGES + 2\NUMENG}$, there are
instances with PoA $\Omega(\NUMENG)$.
\end{proposition}

\begin{proof}
Let \TypeDist be a nearly uniform distribution in general position;
specifically, for each page $n$, assume
$1/\NUMPAGES - 1/\NUMPAGES^2 \leq \Td{n} \leq 1/\NUMPAGES + 1/\NUMPAGES^2$.
We will reuse the calculations from the proof of
Theorem~\ref{thm:intermediate.sym} to first find the symmetric
strategies that form a local maximum of the Lagrangians $L_i$,
then to verify that they constitute an equilibrium.

We begin by lower-bounding the value of $\lambda'_i$
from Equation~\eqref{eqn:lambda-condition}.
By the choice of \Td{n}, we have
\begin{align*}
\NUMPAGES \frac{1/\NUMPAGES-1/\NUMPAGES^2}{%
\lambda'_i - z(1/\NUMPAGES-1/\NUMPAGES^2)}
& \leq \sum_n \frac{\Td{n}}{\lambda'_i - z \Td{n}}
\; = \; 1,
\end{align*}
which can be solved for $\lambda'_i$ to conclude that
\begin{align*}
\lambda'_i
& \geq (\NUMPAGES+z) \cdot (1/\NUMPAGES - 1/\NUMPAGES^2).
\end{align*}

We next substitute this bound into Equation~\eqref{eqn:beta-symmetric-response}.
Notice that $\frac{\Td{n}}{\lambda'_i - z \Td{n}}$ is monotone
decreasing in $\lambda'_i$ and monotone increasing in \Td{n}. 
Therefore, 
\begin{align*}
\SPr{n}
& \leq
\frac{1/\NUMPAGES + 1/\NUMPAGES^2}{%
(\NUMPAGES+z) \cdot (1/\NUMPAGES - 1/\NUMPAGES^2)
- z(1/\NUMPAGES + 1/\NUMPAGES^2)}
\; = \;
\frac{1/\NUMPAGES+1/\NUMPAGES^2}{1-1/\NUMPAGES-2z/\NUMPAGES^2}
\; \leq \; \frac{2}{\NUMPAGES}
\end{align*}
for sufficiently large \NUMPAGES.
(We used 
$z = \frac{(2\NUMENG-1)(1-\DF)}{(\NUMENG-1) \DF} \leq 2\NUMPAGES$ 
for $\DF > \frac{2\NUMENG}{\NUMPAGES + 2\NUMENG}$.)
We have shown that the unique symmetric local maximum is nearly
uniform (up to lower-order terms).

To verify that the \SPr{n} indeed constitute an equilibrium,
we verify that the best-response function for engine $i$ is strictly
concave, which implies that the local maximum is a global maximum.
It is enough to verify that the right-hand side of
Equation~\eqref{eqn:hessian-diagonal} is negative, since this implies
negative definiteness of the Hessian.
To see this, we observe that at the symmetric strategy profile,
$\sum_{k \neq i}\SP{k}{n} \leq \frac{2(\NUMENG-1)}{\NUMPAGES}$, which
implies that $\DF-(1-\DF) \cdot \sum_{k \neq i}\SP{k}{n} \geq 0$,
and thus that
$\frac{d^2 \Payoff{i}{\SatProb{i},\SatProbV{-i}}}{d^2 \SP{i}{n}} < 0$.

Finally, having shown the existence of a symmetric equilibrium, we
infer that the PoA is $\Omega(\NUMENG)$.
At equilibrium, a searcher can do no better than to pick an engine
arbitrarily, yielding a welfare of
$\sum_n \Td{n} \cdot \SPr{n} \leq \sum_n \Td{n} \cdot 2/\NUMPAGES
= 2/\NUMPAGES$.
On the other hand, in the optimum strategy, each engine $i$ would
deterministically display page $i$ (without loss of generality, assume
that the pages are sorted by decreasing \Td{n}), yielding a searcher
welfare of at least $\NUMENG/\NUMPAGES$.
\end{proof}

\begin{proposition}
\label{prop:intermediate-sqrt}
Consider the \SingletonGameName with proportional selection rule
\SelRule. 
If $\DF > \frac{6}{\sqrt{\NUMPAGES}}$, there are
instances with $\NUMENG = 2$ and PoA $\Omega(\sqrt{\NUMPAGES})$. 
\end{proposition}

\begin{proof}
We use a distribution \TypeDist in general position falling into the
class used in the second part of the proof of Proposition~\ref{prop:poa.n}.
The first two pages are requested frequently:
$1/\sqrt{\NUMPAGES} \leq \Td{1}, \Td{2} \leq 2/\sqrt{\NUMPAGES}$,
while the remaining pages are requested roughly uniformly:
$1/(\NUMPAGES-2) - 4/\NUMPAGES^{1.5} \leq \Td{n} \leq 2/\NUMPAGES$.
(It is easy to verify that these constraints admit distributions.)
The outline of the proof is the same as for
Proposition~\ref{prop:intermediate-engines}: we use the
characterization of $\lambda'_i$ and \SPr{n} from the proof of
Theorem~\ref{thm:intermediate.sym} to lower-bound $\lambda'_i$,
then use this lower bound to upper-bound \SPr{1} and \SPr{2}
by $O(1/\sqrt{\NUMPAGES})$.
Using these bounds, we can verify that the symmetric
local maximum of the Lagrangian is an equilibrium, and bound the
resulting PoA. The calculations are more
involved due to the lack of symmetry across pages.
First, from Equation~\eqref{eqn:lambda-condition}, we get
\begin{align*}
1 
& = \sum_n \frac{\Td{n}}{\lambda'_i - z \Td{n}} \\
\; & \geq \; 
   \frac{2/\sqrt{\NUMPAGES}}{\lambda'_i - z/\sqrt{\NUMPAGES}} 
   + \frac{1-4(\NUMPAGES-2)/\NUMPAGES^{1.5}}{%
   \lambda'_i - z \cdot (1/(\NUMPAGES-2) - 4/\NUMPAGES^{1.5})} \\
\; & \geq \;
   \frac{2/\sqrt{\NUMPAGES}}{\lambda'_i - z/\sqrt{\NUMPAGES}} 
   + \frac{1-4/\sqrt{\NUMPAGES}}{\lambda'_i}.
\end{align*}
Multiplying through and rearranging yields
\begin{align*}
(\lambda'_i-1) (\lambda'_i-z/\sqrt{\NUMPAGES})
& \geq
4z/\NUMPAGES - 2\lambda'_i/\sqrt{\NUMPAGES} 
\; \geq \;
2/\sqrt{\NUMPAGES} \cdot (z/\sqrt{\NUMPAGES} - \lambda'_i).
\end{align*}
In the proof of Theorem~\ref{thm:intermediate.sym}, we argued that
$\lambda'_i > z \Td{n}$ for all $n$, so
$\lambda'_i-z/\sqrt{\NUMPAGES} > 0$, and we can divide by it to obtain
that $\lambda'_i \geq 1 - 2/\sqrt{\NUMPAGES}$.

We next use this lower bound on $\lambda'_i$,
and the upper bound that $\Td{n} \leq 2/\sqrt{\NUMPAGES}$,
to upper-bound the probability \SPr{n} with which engines display pages.
Substituting the bounds into
Equation~\eqref{eqn:beta-symmetric-response}, we obtain
\begin{align*}
\SPr{n} 
& \leq \frac{2/\sqrt{\NUMPAGES}}{1-2/\sqrt{\NUMPAGES} - z/\sqrt{\NUMPAGES}}
\; = \; \frac{2}{\sqrt{\NUMPAGES} - 2 - z}
\; \leq \; \frac{4}{\sqrt{\NUMPAGES}},
\end{align*}
because we assumed that 
$\DF \geq 6/\sqrt{\NUMPAGES} \geq 6/(2+\sqrt{\NUMPAGES})$, 
which implies that $z \leq \sqrt{\NUMPAGES}/2-2$.

To verify that the Hessian is negative definite, we show that the
right-hand side of Equation~\eqref{eqn:hessian-diagonal} is negative.
Notice that 
$\sum_{k \neq i}\SP{k}{n} \leq 4/\sqrt{\NUMPAGES}$ for all pages $n$,
implying $\DF-(1-\DF) \cdot \sum_{k \neq i}\SP{k}{n} \geq 0$,
because 
$\DF \geq 6/\sqrt{\NUMPAGES} \geq 4/(4+\sqrt{\NUMPAGES})$.
Thus, $\frac{d^2 \Payoff{i}{\SatProb{i},\SatProbV{-i}}}{d^2 \SP{i}{n}} < 0$.

To bound the welfare at this equilibrium, note the
probability that a user is satisfied is
\begin{align*}
\sum_n \Td{n} \SPr{n} 
& \leq 
2 \cdot 2/\sqrt{\NUMPAGES} \cdot 4/\sqrt{\NUMPAGES}
+ 2/\NUMPAGES \sum_{n>2} \SPr{n}
\; = \; O(1/\NUMPAGES).
\end{align*}
On the other hand, the optimum solution would have engine $i=1,2$
deterministically display page $i$, for a welfare of
at least $2/\sqrt{N}$. The PoA is thus $\Omega(\sqrt{N})$.
\end{proof}


\section{Convexity of Markov Rules}
\label{sec:markov}


We now describe a natural model of user behavior under which the
selection rule is always convex.
By Theorems~\ref{thm:convex.deterministic} and \ref{thm:bounded-PoA-for-deterministic}, 
if users behave according to this model, searcher welfare is high.
The model generalizes the Markov selection rule described in
Section~\ref{sec:models}.  
Each search engine $i$ corresponds to a unique state.
The user performs a walk on states, and always uses the engine
corresponding to his current state.
The transition probabilities between states depend on whether the user
was satisfied in the previous round; 
we use \TransSucc{i}{j} and \TransFail{i}{j} to denote the probability of
moving from state $i$ to state $j$ upon satisfaction or failure
(non-satisfaction), respectively.
The process has a stationary distribution \StatProb, and the 
induced \emph{Markovian selection rule} is 
$\SelRule[i](\SatProbs) = \StatProb[i]$.

Leaving state $i$ after a successful (resp., failed) query induces a
distribution over the state of the process in the next round.
In turn, that state defines an expected return time \RetSucc{i} (resp.,
\RetFail{i}) to state $i$: the expected number of rounds before the process
returns to state $i$. Note that $\RetSucc{i}, \RetFail{i} \geq 1$.

We call a Markov process \emph{monotone} if, for all $i$, 
$\TransSucc{i}{i} \geq \TransFail{i}{i}$ and $\RetSucc{i} \leq \RetFail{i}$. 
The first condition states that a user is no more likely to leave his
current engine after a successful query than after a failed query.
The second condition says that, conditioning on having left engine
$i$, the last query having been successful should not increase the
expected return time.  The Markov process is \emph{strictly
  monotone} if the first inequality is strict. 
We show that monotone Markovian selection rules are convex.  

\begin{theorem}
\label{thm:markov}
Any (strictly) monotone Markovian selection rule is (strictly) convex.
\end{theorem}

\begin{proof}
For a given Markov process, consider the time spent in state $i$.
Write $\ExitFail{i} = 1-\TransFail{i}{i}$ and $\ExitSucc{i} = 1-\TransSucc{i}{i}$
for the probability that the process leaves state $i$ following a
failure or success, respectively.  
Given that the process is in state $i$, the expected number of rounds
until the process leaves state $i$ is 
$\frac{1}{(1-\SatProb{i})\ExitFail{i} + \SatProb{i} \ExitSucc{i}}$.
Conditioned on having left state $i$, the
probability that the last search was successful is
$\frac{\SatProb{i} \ExitSucc{i}}{(1-\SatProb{i})\ExitFail{i} + \SatProb{i} \ExitSucc{i}}$,
and the probability that it failed is 
$\frac{\SatProb{i} \ExitSucc{i}}{(1-\SatProb{i})\ExitFail{i} + \SatProb{i} \ExitSucc{i}}$.
We conclude that the
expected return time, conditioned on having left state $i$, is
$\frac{\SatProb{i} \ExitSucc{i}}{(1-\SatProb{i})\ExitFail{i} + \SatProb{i} \ExitSucc{i}}\RetSucc{i} +
\frac{(1-\SatProb{i}) \ExitFail{i}}{(1-\SatProb{i})\ExitFail{i} + \SatProb{i} \ExitSucc{i}}\RetFail{i}$.
Since the fraction of time spent in state $i$ is the expected stay
length divided by expected stay length plus return time, we have that
\begin{align*} 
\StatProb[i]
& = 
\frac{\frac{1}{(1-\SatProb{i})\ExitFail{i} + \SatProb{i} \ExitSucc{i}}}{\frac{1}{(1-\SatProb{i})\ExitFail{i} + \SatProb{i} \ExitSucc{i}} + \frac{\SatProb{i} \ExitSucc{i}}{(1-\SatProb{i})\ExitFail{i} + \SatProb{i} \ExitSucc{i}}\RetSucc{i} + \frac{(1-\SatProb{i})\ExitFail{i}}{(1-\SatProb{i})\ExitFail{i} + \SatProb{i} \ExitSucc{i}} \RetFail{i}} \\
& = \frac{1}{1 + \SatProb{i} \ExitSucc{i} \RetSucc{i} + (1-\SatProb{i}) \ExitFail{i} \RetFail{i}}.
\end{align*}
We must show that \StatProb[i] is a valid selection rule, meaning
that it is non-decreasing in \SatProb{i}.
Taking the derivative with respect to \SatProb{i}, we have 
\begin{align*}
\frac{d}{d \SatProb{i}}\StatProb[i] 
& = \frac{\ExitFail{i} \RetFail{i} - \ExitSucc{i} \RetSucc{i}}{(1+\ExitFail{i}\RetFail{i}(1-\SatProb{i})+\ExitSucc{i}\RetSucc{i} \SatProb{i})^2},
\end{align*}
which is non-negative since $\ExitFail{i} \geq \ExitSucc{i}$ and $\RetFail{i} \geq \RetSucc{i}$.  
Next, we show that \StatProb[i] is a convex function of $\SatProb{i}$.  
Taking the second derivative with respect to \SatProb{i}, we have
\begin{align*}
\frac{d^2}{d \SatProb{i}^2} \StatProb[i] 
& = \frac{2(\ExitFail{i} \RetFail{i} - \ExitSucc{i}
  \RetSucc{i})^2}{(1+\ExitFail{i}\RetFail{i}(1-\SatProb{i})+\ExitSucc{i}\RetSucc{i} \SatProb{i})^3},
\end{align*}
which is non-negative, and strictly positive whenever 
$\ExitFail{i} > \ExitSucc{i}$.
\end{proof}

\section{Conclusions and Future Work}
\label{sec:conclusions}

We introduced a natural model of competition between multiple search
algorithms, who are vying for visits from users searching for web pages.
Our model interpolates between different objectives for the search engines,
and allows different user behavior models.
At one extreme is the
pay-per-impression advertising, corresponding to an
objective of attracting as many searchers as possible.
At the other extreme is 
pay-per-click advertising, corresponding to
an objective of attracting \emph{and satisfying} as many searchers
as possible.

Our main result is a strong dichotomy: when engines are directly incentivized
to satisfy users, 
we expect to see specialization,
leading to a market that serves users well and achieves a
(tight) Price of Anarchy of 2.
On the other hand, with pure pay-per-impression advertising,
search engines will generally play symmetrically, duplicating each
other's strategies.
The Price of Stability in this setting can be as bad as
$\Omega(\NUMENG)$ or $\Omega(\sqrt{\NUMPAGES})$.
Furthermore, the bad equilibria often persist when search
engines interpolate between these extremes, even
when the fraction of pay-per-impression advertising is quite small.
However, these pessimistic examples vanish under natural models of
user behavior in which the demand for an engine is a convex
function of that engine's search result quality.


A number of questions remain open.
Our lower bounds apply to the pure Price of Stability; we did not rule
out the existence of more benign mixed (or
correlated) equilibria. Can the lower bounds be generalized?
Our understanding of equilibria for $\DF \in (0,1)$ is even more
limited; how can we characterize the equilibria in this
case, and can we rule out asymmetric pure equilibria?

One could also consider approximate equilibria, in which
engines cannot unilaterally deviate to improve their payoffs \emph{much}.
It would be interesting to study the Price of Anarchy/Stability
for approximate equilibria.

Finally, another layer of complexity may be added to the model 
to be more realistic. Our model posited that 
a user would choose an engine $i$ for each search independently.
In reality, most users are loyal to one search engine,
perhaps due to inertia.
This could be modeled by assuming that each user has a
distribution over potential queries, and picks one engine to maximize
the expected probability of satisfaction across all queries.
Analyzing the Price of Anarchy/Stability of this more complex game is
an interesting direction for future work.

\subsubsection*{Acknowledgments}
We thank Shaddin Dughmi, Matt Elliott, Xinran He and Nicole Immorlica
for useful discussions.
Work done in part while David Kempe was visiting Microsoft Research,
New England.

\bibliographystyle{abbrv}
\bibliography{../bibliography/names,../bibliography/conferences,../bibliography/bibliography,../bibliography/publications,../bibliography/bibliography-local}

\appendix
\section{Non-existence of Pure Equilibria}
\label{app:non-existence}
We justify the broadened focus on the correlated Price of Anarchy by
showing here that even the \SingletonGameName with $\DF=0$ may not
have any pure Nash Equilibria.
We construct an instance with $\NUMPAGES=2$ pages and $\NUMENG=3$
engines. The page request probabilities are in general position,
with $\Td{n} \in (\third, \frac{2}{3})$ for $n=1,2$.  
The selection rule is as follows.
For given $\SatProbs$, define 
$\gamma_i = \SatProb{i} \cdot 2^{(\SatProb{i+1} - \SatProb{i+2})}$, 
with indices taken mod $3$. 
The selection rule is then 
$\SelRule[i](\SatProbs) = \gamma_i/(\gamma_1 + \gamma_2 + \gamma_3)$.  
For example, $\SelRule[1](1,1,0) = 4/5$, 
whereas $\SelRule[1](1,0,1) = 1/5$.
Note that the \SelRule[i] satisfy the non-indifference property. 

By Theorem~\ref{thm:convex.deterministic}, if a pure equilibrium exists, it must
be a profile of deterministic strategies.  There are two cases.  
First, if all three engines select the same page $n$, then the
utility of each engine is 
$\Td{n} \cdot \SelRule[i](1,1,1) = \Td{n} \cdot \third < 2/9$.
In this case, each engine would benefit by switching to the other
page, which would give a utility of $\third > 2/9$. 
Thus, this profile cannot be an equilibrium.

The second case is that two engines select one page, and the third
engine selects the other.
By symmetry, we can assume that engines $1$ and $2$ select the same
page $n$.  In this case, the utility of engine $2$ is 
$\Td{n} \cdot \SelRule[2](1,1,0) = \Td{n} \cdot \frac{1}{5} \leq 2/15$.
Note, then, that engine $2$ would benefit by switching to the other
page $n'$; her utility would become 
$\Td{n'} \cdot \SelRule[2](0,1,1) = \Td{n'} \cdot \frac{4}{5} > 4/15$.
This profile is therefore not an equilibrium, either.
As we have exhausted all possibilities for a pure equilibrium, 
we conclude that no pure equilibrium exists.

\section{Necessity of Assumptions in Theorem~\ref{thm:convex.deterministic}}
\label{app:necessity}
At first sight, the assumptions of non-indifference of the \SelRule[i]
and general position of \TypeDist in Theorem \ref{thm:convex.deterministic} may 
appear like minor technical issues that should be easily eliminated. 
However, we provide examples showing that the bound in 
Theorem~\ref{thm:convex.deterministic} requires both assumptions, even for
instances of the \SingletonGameName.

\begin{example}[Necessity of Non-Indifference]
Consider an instance of the \SingletonGameName with \NUMPAGES pages
and $\NUMENG = \NUMPAGES+1$ engines. 
Choose any distribution \TypeDist in general position over these
singleton sets, subject to the constraint that 
$\Td{n} \in [\frac{2}{3\NUMPAGES}, \frac{4}{3\NUMPAGES}]$
for each page $n$.  

Consider the following selection rule.
First, engine $\NUMPAGES+1$ is selected with probability
$\SelRule[\NUMPAGES+1](\SatProbs) 
= 1 - \sum_{i \leq \NUMPAGES} \SelRule[i](\SatProbs)$, 
so it suffices to describe the selection rule on engines 
$i \leq \NUMPAGES$.
For each $i \leq \NUMPAGES$, if engine $i$ is the only engine in $[\NUMPAGES]$ 
with non-zero probability of displaying page $n$
(i.e., $q_i > 0$ and $q_j = 0$ for all $j \in [\NUMPAGES]\backslash\{i\}$),
then the user will always choose engine $i$, so $\SelRule[i](\SatProbs) = 1$.
Otherwise, let 
$\SelRule[i](\SatProbs) = \max(0, \min(\frac{1}{\NUMPAGES}, \half(1 - \sum_{j \in [\NUMPAGES]\backslash\{i\}} q_j)))$.
In other words, engine $i$ is selected with probability
$\half(1 - \sum_{j \in [\NUMPAGES]\backslash\{i\}} q_j)))$, but the value is
``truncated'' into the interval $[0,\frac{1}{\NUMPAGES}]$.  
Note that the \SelRule[i] are not non-indifferent, 
since $\SelRule[i](\SatProbs)$ does not depend on \SP{i}{n}
when $\SP{j}{n} \neq 0$ for any $j \in [\NUMPAGES]\backslash\{i\}$.
The selection rules \SelRule[i] are, however, non-decreasing in
\SP{i}{n} and non-increasing in \SP{j}{n} for $j \neq i$.

The socially optimal strategy profile is for each engine 
$i \leq \NUMPAGES$ to deterministically display page $i$. 
This profile leads to a welfare of $1$. 

We now describe an equilibrium profile.
Let $X = \big(\sum_{j  \in \PageSet} \frac{1}{\Td{j}} \big)^{-1}$.
Consider the symmetric profile in which every engine selects page $n$ with
probability $\rho_n = \frac{\Td{n}-X}{\Td{n}(\NUMPAGES-1)}$.
Note that since $\Td{n} \in [\frac{2}{3\NUMPAGES}, \frac{4}{3\NUMPAGES}]$ 
for each $n$, we have 
\begin{align*}
0 
& < \rho_n 
\; = \; \frac{1-\frac{X}{\Td{n}}}{\NUMPAGES-1} 
\; \leq \; \frac{1 - \frac{3\NUMPAGES/4}{3\NUMPAGES^2/2}}{\NUMPAGES-1}
\; \leq \; \frac{2}{\NUMPAGES}.
\end{align*}
Moreover, the values $(\rho_n)$ form a probability distribution (and
hence a valid strategy profile) since  
\begin{align*}
\sum_{n \in \PageSet} \rho_n 
& = \frac{1}{\NUMPAGES-1} \cdot \sum_j 
\left( 1 - \frac{1/\Td{j}}{\sum_{n \in \PageSet} 1/\Td{n}}\right) 
\; = \; \frac{1}{\NUMPAGES-1} \cdot
\left(\NUMPAGES - \frac{\sum_{j\in \PageSet} 1/\Td{j}}{\sum_{n \in \PageSet} 1/\Td{n}}\right)
\; = \; 1.
\end{align*}
Under this strategy profile, for engine $i \in [\NUMPAGES]$ and page
$n$, we have that
\begin{align*}
\SelRule[i](\SatProbs) 
& = \max\left(0, \min\left(\frac{1}{\NUMPAGES}, \half (1 - (\NUMPAGES - 1)\rho_n)\right)\right)
\; = \; \min\left(\frac{1}{\NUMPAGES}, \frac{X}{2\Td{n}}\right).
\end{align*}
Because $\Td{n'} \in [\frac{2}{3\NUMPAGES}, \frac{4}{3\NUMPAGES}]$ for
all pages $n'$, we get that $\Td{n} \leq 2X/\NUMPAGES$, which in turn
implies that $\frac{X}{2\Td{n}} \leq \frac{1}{\NUMPAGES}$,
so $\SelRule[i](\SatProbs) = \frac{X}{2\Td{n}}$.
The utility of engine $i \leq \NUMPAGES$ under this strategy profile
is therefore
\begin{align*}
\sum_n \Td{n} \SelRule[i](\SatProbs(n)) \cdot \SP{i}{n}
& = \sum_{n \in \PageSet} X/2 \cdot q_i(n) 
\; = \; X/2,
\end{align*}
which is independent of her strategy.
Thus, every strategy is a best-response for engine $i \leq \NUMPAGES$;
in particular, the proposed profile constitutes an equilibrium
for the engines $i \leq \NUMPAGES$.
Since the strategy of engine $\NUMPAGES+1$ does not influence the
payoff of the other engines, we may simply assume that engine
$\NUMPAGES+1$ plays some best response.

To bound the welfare of the equilibrium, we consider separately the
engines $i \leq \NUMPAGES$ and engine $\NUMPAGES+1$.
Engine $\NUMPAGES+1$ can satisfy a user with probability at most
$\max_n \SET{\Td{n}} = O(1/\NUMPAGES)$. 
For engines $i \leq \NUMPAGES$, because their strategies are
symmetric, the social welfare is precisely the probability that an
arbitrary engine satisfies a user type drawn from \TypeDist. 
This is at most 
$\NUMPAGES \cdot \max_n \SET{\rho_n} \cdot \max_n \SET{\Td{n}} =
O(1/\NUMPAGES)$.
The Price of Anarchy in this example is therefore $\Omega(\NUMPAGES)$. 
\end{example}

\begin{example}[Necessity of General Position]
Consider an instance of the \SingletonGameName with \NUMENG engines
and $\NUMPAGES = \NUMENG$ pages. 
The distribution over desired pages is uniform:
$\Td{n} = 1/\NUMPAGES$ for each page $n$.
The selection function is a weighted proportional rule: 
$\SelRule[i](\SatProbs) = \frac{w_i \SatProb{i}}{\sum_j w_j \SatProb{j}}$,
where $w_1 = \NUMENG^2$ and $w_j = 1$ for all $j \neq 1$. 
(In other words, users have a significant preference for engine 1.)

The socially optimal strategy profile is for each engine $i$ to
deterministically display page $i$ in her first slot; this leads to a
social welfare of $1$.
Consider the following strategy profile: 
engine $1$ selects one of pages $1$ through $\NUMENG-1$ uniformly at
random to display in her first slot, 
and all remaining engines deterministically display page \NUMENG in
their first slot with probability $1$.
We claim that this profile constitutes an equilibrium. 
To see why, note that engine $1$ does not compete with any other
engine for pages $1$ through $\NUMENG - 1$.
Thus, for any distribution $p$ over pages $1$ through $\NUMENG - 1$ 
chosen by engine $1$ for slot 1, her utility is 
$\sum_{n < \NUMENG} \frac{1}{\NUMENG} p(n) = \frac{1}{\NUMENG}$;
on the other hand, any weight placed on page \NUMENG is multiplied
by a factor strictly less than 1 due to the competition, 
and thus leads to a loss in utility.
On the other hand, if any other engine chooses to display in her first
slot a page $n < \NUMENG$, her utility conditioned on selecting that
page is $\frac{1}{\NUMENG} \cdot
\frac{1}{1+\NUMENG^2\frac{1}{\NUMENG-1}} < \frac{1}{\NUMENG^2}$,
which is the utility derived from page \NUMENG.
Thus, each engine $i > 1$ maximizes her payoff by displaying page
\NUMENG in her first slot with probability $1$.
In summary, this strategy profile is an equilibrium.

In this equilibrium, a user who desires page \NUMENG is satisfied
with probability $1$ (since he will choose any of the engines that
deterministically display it in the first slot), 
whereas any other user is satisfied with probability
$\frac{1}{\NUMENG-1}$ (since he must choose engine 
$1$, who displays the desired page with this probability).
The social welfare is thus 
$\frac{1}{\NUMENG} + \frac{\NUMENG-1}{\NUMENG}\cdot
\frac{1}{\NUMENG-1} = \frac{2}{\NUMENG}$;
hence, the Price of Anarchy in this example is $\Omega(\NUMENG)$. 
\end{example}

\end{document}